%% file: kpp.tex
\documentclass{article} 
\usepackage{latexsym, amsmath}
\usepackage{lscape}
\usepackage{epstopdf,caption,subcaption,graphicx,xcolor}
\usepackage{tikz}
\usepackage{hyperref}
\usepackage{comment}

\usetikzlibrary{shapes,arrows,fit,calc,positioning}
\usetikzlibrary{decorations.pathreplacing}

\normalsize
\newtheorem{theorem}{Theorem}
\newtheorem{lemma}{Lemma}

\newtheorem{notation}{Notation}

\newtheorem{operation}{Operation}
\newtheorem{definition}{Definition}

\newtheorem{problem}{Problem}
\newenvironment{proof}{{\sc Proof. }}{\hfill$\Box$\vspace{0.1in}}
\newcommand{\mc}[1]{{\cal {#1}}}

\title{Approximately Partitioning Vertices into Short Paths}
\author{
Mingyang~Gong\thanks{Gianforte School of Computing, Montana State University, Bozeman, MT 59717, USA.
Email: {\tt \{mingyang.gong, brendan.mumey\}@montana.edu}}
\and
Zhi-Zhong~Chen\thanks{Division of Information System Design, Tokyo Denki University. Saitama 350-0394, Japan.
  Email: \texttt{zzchen@mail.dendai.ac.jp}}
\thanks{Correspondence authors.}
\and
Brendan Mumey$^*$$^\ddagger$
}
\date{}
\begin{document}
\maketitle
\begin{abstract}
Given a fixed positive integer $k$ and a simple undirected graph $G = (V, E)$, 
the {\em $k^-$-path partition} problem, denoted by $k$PP for short,
aims to find a minimum collection $\mc{P}$ of vertex-disjoint paths in $G$
such that each path in $\mc{P}$ has at most $k$ vertices and
each vertex of $G$ appears in one path in $\mc{P}$.
In this paper, we present a $\frac {k+4}5$-approximation algorithm for $k$PP when $k\in\{9,10\}$
and an improved $(\frac{\sqrt{11}-2}7 k + \frac {9-\sqrt{11}}7)$-approximation algorithm when $k \ge 11$.
Our algorithms achieve the current best approximation ratios for $k \in \{ 9, 10, \ldots, 18 \}$.

Our algorithms start with a maximum triangle-free path-cycle cover $\mc{F}$, which may not be feasible because of the existence of 
cycles or paths with more than $k$ vertices.
We connect as many cycles in $\mc{F}$ with $4$ or $5$ vertices as possible by 
computing another maximum-weight path-cycle cover in a suitably constructed graph 
so that $\mc{F}$ can be transformed into a $k^-$-path partition of $G$ without losing too many edges.

\end{abstract}

\section{Introduction}\label{sec:intro}

Throughout this paper, we assume that $k \ge 1$ is a fixed positive integer and 
the input $G = (V, E)$ is always a simple undirected graph where $V$ and $E$ are the set of vertices and edges of $G$, respectively.
Given a path (respectively, cycle) $P$ in $G$, the {\em order} $\ell$ of $P$ is the number of vertices in $P$
and we call $P$ an $\ell$-path (respectively, $\ell$-cycle).
Similarly, an $\ell^-$-path (respectively, $\ell^+$-path) is a path having order at most (respectively, at least) $\ell$.
A {\em $k^-$-path partition}, denoted by $k$-pp for short, of $G$ is a set of vertex-disjoint $k^-$-paths in $G$
such that every vertex in $G$ appears in one of the paths in the set.
The {\em $k^-$-path partition problem}, denoted by $k$PP, seeks to find a $k$-pp of $G$ whose size is minimized.
$k$PP has various real-life applications in vehicle routing, network monitoring and others~\cite{YCH97, S03, CCK24, LYL24b}.
We study $k$PP from the perspective of approximation algorithms.
An algorithm is {\em $\rho$-approximation} for a minimization (respectively, maximization) problem $\Pi$
if it runs in a polynomial time in the input size
and outputs a feasible solution that is within (respectively, at least) a guaranteed factor $\rho$ of the optimal solution.
Therefore, $\rho \ge 1$ if $\Pi$ is a minimization problem and otherwise, $\rho \le 1$.

We note that $1$PP is trivial.
Moreover, $2$PP is equivalent to finding a maximum matching of $G$ and hence can be solved in $O(\sqrt{|V|} |E|)$ time~\cite{MV80}.
Unfortunately $k$PP becomes NP-hard~\cite{GJ79} when $k \ge 3$. 
On the other hand, $k$PP can be solved in polynomial time on some special graphs such as
trees~\cite{YCH97}, cographs~\cite{S00} and bipartite permutation graphs~\cite{S03}.

\subsection{Previously known results for $k$PP}
We here review the algorithmic results for $k$PP with $k \ge 3$ on a general graph.

The first group of algorithms are designed for a general $k \ge 3$.
In~\cite{CCL22, CGL22}, the authors observed that minimizing the number of $1$-paths, i.e., singletons in a $k$-pp
is polynomial-time solvable, which results in a $\frac k2$-approximation algorithm.
The ratio was slightly improved to $\frac {k^2}{3(k-1)} + \frac {k-3}{6(k-1)} \cdot (k \mod 3)$~\cite{LYL24a}
by merging two short paths into a single path or; into another two paths, one of which is a $k$- or $(k-1)$-path.
Note that the number of paths in a $k$-pp of $G$ is equal to $|V|$ minus 
the total number of edges in the paths in the $k$-pp.
Therefore, $k$PP is close to the {\em maximum path cover} problem~\cite{BK06}, denoted by MPC for short,
which aims to find a set of vertex-disjoint paths such that the total number of edges is maximized. 
Indeed, as observed by Li et al.~\cite{LYL24b}, an obvious approximation algorithm for $k$PP can use the $\frac 67$-approximation algorithm for MPC~\cite{BK06} 
to obtain a set of vertex-disjoint paths and then cut every $(k+1)^+$-path among them into a set of vertex-disjoint $k^-$-paths. 
This simple algorithm achieves an approximation ratio of $\frac {k+12}7-\frac 6{7k}$ \cite{LYL24b}.
Recently, the authors in~\cite{IGK25} presented the current best $(\frac {k+12}7-\frac 1{k-1})$-approximation algorithm for a general $k \ge 9$.
If $k$ is part of the input, then when $k = |V|$, $k$PP is reduced to the {\em path partition} problem~\cite{FR02},
which in turn includes the classic {\em Hamiltonian path} problem~\cite{GJ79} as a special case.
Li et al.~\cite{LYL24b} showed that $k$PP cannot be approximated within $O(k^{1-\epsilon})$ for any $\epsilon>0$ unless NP=P.
We remark that the authors in~\cite{CCK24} extended $k$PP to directed graphs and presented a $\frac k2$-approximation algorithm for $k \ge 3$
and an improved $\frac {k+2}3$-approximation algorithm for $k \ge 7$.

The second group of algorithms are presented for $3$PP.
Monnot and Toulouse~\cite{MT07} presented the first matching-based $\frac 32$-approximation algorithm for $3$PP.
Chen et al.~\cite{CGL22} improved the ratio to $\frac {13}9$ by computing a $3$-pp with the least number of $1$-paths and 
merging three $2$-paths into two $3$-paths.
Later on, Chen et al.~\cite{CCL22} presented a $\frac 43$-approximation algorithm that refines the current solution
by checking more $2$- and $3$-paths.
Finally, the authors in~\cite{CGS19} showed the current best $\frac {21}{16}$-approximation algorithm for 3PP,
which applies more elaborate local search operations on a weighted auxiliary graph.

The last group of algorithms are designed for some specific values of $k$.
Li et al.~\cite{LYL24b} developed a local search based algorithm for $4$PP, $5$PP, and $6$PP,  
achieving a ratio of $\frac {31}{18}$, $\frac {17}8$, and $\frac 73$, respectively. 
Those ratios were improved to $\frac 85$, $\frac {55}{31}$ and $2$~\cite{GCC25}, respectively
and the authors in~\cite{GCC25} also presented the current best $\frac 73$- and $\frac 83$-approximation for $7$PP
and $8$PP, respectively.

A related concept to $k$PP is a {\em path-cycle cover} in $G$,
which is a spanning subgraph of $G$ such that each connected component is a path or cycle.
Since the degree of each vertex in a path-cycle cover is at most $2$, a path-cycle cover is also called a {\em $2$-matching}.
A $3$-cycle is called a {\em triangle}.
A {\em triangle-free path-cycle cover} is a path-cycle cover without any triangles. 
A {\em (triangle-free) path-cycle cover} is {\em maximum} if its number of edges is maximized.
Fortunately, a maximum path-cycle cover can be found in $O(|E||V| \log |V|)$ time~\cite{Gab83} 
while a maximum triangle-free path-cycle cover can be found in $O(|V|^3 |E|^2)$ time~\cite{Har24}.

Another related concept to $k$-pp is a {\em $2$-piece packing} in $G$~\cite{HHS06}, 
which is a subgraph in which each connected component is a {\em 2-piece}, i.e., a $3^+$-path or cycle.
A {\em maximum $2$-piece packing} is a $2$-piece packing whose number of vertices is maximized.
The authors in~\cite{HHS06} presented a polynomial-time algorithm for computing a maximum $2$-piece packing.

Note that in a $k$-pp of $G$, each path should be a $k^-$-path, i.e., a ``short'' path.
Gong et al.~\cite{GEF24} proposed a complementary problem of $k$PP that aims to cover as many vertices 
as possible by a collection of vertex-disjoint {\em $k^+$-paths}.
The problem is denoted by MPC$^{k+}$ for simplicity.
One sees that a $(2k)^+$-path can be cut into two $k^+$-paths and 
thus, we can assume that in the target solution of MPC$^{k+}$, the order of each path is in between $k$ and $2k-1$.
Therefore, MPC$^{1+}$ is trivial.
MPC$^{2+}$ is tractable too since it is equivalent to finding a $3$-pp with the least number of $1$-paths~\cite{CCL22}.
Note that a $3^+$-path is a $2$-piece and a cycle becomes a $3^+$-path with the same order after the removal of an arbitrary edge.
One clearly sees that MPC$^{3+}$ is the same as the maximum $2$-piece packing problem.
Thus, MPC$^{3+}$ is tractable~\cite{HHS06}.
However, the authors in~\cite{KLM23} proved that MPC$^{k+}$ is NP-hard for every $k \ge 4$
and they also presented a $0.25$-approximation algorithm for $4$PP.
Gong et al.~\cite{GEF24} improved the ratio to $0.5$ by presenting five local search operations,
each of which increases the number of covered vertices or the number of $4$-paths in the current solution.
Later on, Gong et al.~\cite{GCL25a} proposed a $0.533$-approximation algorithm for $4$PP by a completely different method,
which computes a maximum matching~\cite{MV80} and then transform as many edges in the matching into a feasible solution as possible.
The current best ratio is $0.6$~\cite{GCL25}, which is shown by a novel amortization scheme.
For a general $k \ge 4$, the authors in~\cite{GEF24} presented the first and the current best $\frac 1{0.4394k+O(1)}$-approximation algorithm for MPC$^{k+}$.
Specifically, it is a $\frac 7{19}$-approximation for MPC$^{5+}$ and 
the authors in~\cite{GCL24} improved the ratio to $0.398$ by non-trivially extending the techniques in~\cite{GCL25a}.

\subsection{Our contribution}

In this paper, we contribute a new $\frac {k+4}5$-approximation algorithm for $k$PP for $k \in \{ 9, 10 \}$
and a $(\frac {\sqrt{11}-2}7 k + \frac {9-\sqrt{11}}7)$-approximation algorithm for $k \ge 11$.
Our algorithms achieve the current best approximation ratios when $k \in \{ 9, 10, \ldots, 18 \}$.
In Table~\ref{tab01}, we compare our approximation ratios with the previous best ones in~\cite{IGK25}.

\begin{table}
\caption{Comparing our approximation algorithms against the previous bests}
\label{tab01}
\centering
\begin{tabular}{|c|c|c|c|c|c|c|c|c|c|c|c|c|}
\hline
$k$ & $9$  & $10$ & $11$ & $12$ & $13$ & $14$ & $15$ & $16$ & $17$ & $18$  \\
\hline
Our algorithms & $2.600$  & $2.800$ & $2.881$ & $3.069$ & $3.257$ & $3.445$ & $3.633$ & $3.821$ & $4.009$ & $4.198$ \\
\hline
The ones in~\cite{IGK25} & $2.875$  & $3.032$ & $3.186$ & $3.338$ & $3.488$ & $3.637$ & $3.786$ & $3.933$ & $4.080$ & $4.227$ \\
\hline
\end{tabular}
\end{table}

An important observation is that 
given a $k$-pp of $G$, the number of paths is the same as $|V|$ minus the number of edges.
Based on the observation, we propose a new problem, denoted by $k$PPE in Problem~\ref{pro02}, 
which asks for a $k$-pp whose number of edges is maximized.
In other words, we switch the objective to maximizing the number of edges in a $k$-pp.
We show that an $\alpha$-approximation of $k$PPE is a $((1-\alpha)k+\alpha)$-approximation of $k$PP.
Therefore, in fact, we present a $\frac 45$-approximation algorithm for $k$PPE for $k\in\{9,10\}$
and a $\frac {9-\sqrt{11}}7$-approximation algorithm for $k \ge 11$.

Our algorithms always start with a maximum triangle-free path-cycle cover $\mc{F}$ of $G$. 
Note that an optimal solution of $k$PPE has at most $|E(\mc{F})|$ edges since a $k$-pp must be a triangle-free path-cycle cover.
Therefore, we choose to transform $\mc{F}$ into a $k$-pp.
Obviously, each $4$-cycle in $\mc{F}$ can be turned into a $4$-path by removing one edge but 
the loss in the number of edges is $\frac 14$, which is relatively large.
So, when $k \in \{ 9, 10 \}$, we compute another maximum-weight path-cycle cover $\mc{W}$ in a suitably constructed auxiliary graph
to connect as many $4$-cycles as possible.
This way, we can form a $k$-pp from $\mc{F}$ and $\mc{W}$ without losing many edges in $\mc{F}$.
When $k \ge 11$, the maximum-weight path-cycle cover $\mc{W}$ connects as many $4$- and $5$-cycles as possible.
If we fails to cover a relatively large fraction of the edges in $\mc{F}$, we recursively call the algorithm on a smaller graph,
which finally guarantees a good solution.

The remainder of this paper is organized as follows. 
Section~2 gives some basic definitions.
Section~3 describes the initial common steps of our algorithms.
Section~4 details the $\frac 45$-approximation algorithm for $k$PPE with $k\in\{9,10\}$, while
Section~5 presents the $\frac {9-\sqrt{11}}7$-approximation algorithm with $k \ge 11$.
Finally, Section~6 concludes the paper.

\section{Basic Definitions}\label{sec:def}
We fix a simple undirected graph $G=(V, E)$ and a positive integer $k$ for discussion.
Let $|V| = n$ and $|E| = m$.
The {\em order} of $G$ is the number of vertices in $G$, i.e., $n$.
A graph $G' = (V', E')$ is a {\em subgraph} of $G$ if $V' \subseteq V$ and $E' \subseteq E$.
Furthermore, if $V' = V$ and $E' \subseteq E$, then $G'$ is a {\em spanning subgraph} of $G$.
For simplicity, we also use $V(G')$ and $E(G')$ to denote the set of vertices and edges in $G'$, respectively.

The {\em degree} of a vertex in $G$ is the number of edges incident to it in $G$. 
A {\em star} is a connected subgraph of $G$ such that exactly one vertex, called the {\em center}, 
is of degree at least $2$, and the other vertices, called the {\em satellites}, are of degree exactly $1$. 
A {\em cycle} in $G$ is a connected subgraph such that each vertex has degree exactly $2$
and a $k$-cycle is a cycle with order exactly $k$.
A $3$-cycle is also called a {\em triangle}.
A {\em $k$-path} $P$ of $G$ is a subgraph of $G$ such that the vertices can be ordered as a sequence $v_1, \ldots, v_k$ 
and $E(P) = \{ \{ v_i, v_{i+1} \}: i=1, \ldots, k-1 \}$.
Similarly, a $k^-$-path is a path with order at most $k$.

A {\em path-cycle cover} of $G$ is a spanning subgraph such that each connected component is a path or cycle.
If a path-cycle cover has no triangles, then it is called a {\em triangle-free path-cycle cover}.
A {\em maximum triangle-free path-cycle cover} is one with the maximum number of edges.
In~\cite{Har24}, the author presented an algorithm that computes a maximum triangle-free path-cycle cover in $O(n^3 m^2)$ time.

Given the positive integer $k$, a {\em $k^-$-path partition} of $G$, denoted by $k$-pp for short, 
is a spanning subgraph of $G$ in which each connected component is a $k^-$-path.
Hereafter, by a path in a $k$-pp, we always mean a connected component in the $k$-pp. 
We next formally define $k$PP.

\begin{problem}
\label{pro01}
{\em ($k$PP)}
Given $G$, 
the problem aims to find a {\em minimum} $k$-pp of $G$, 
where {\em ``minimum"} means that the number of paths is minimized.
\end{problem}

Given a $k$-pp $\mc{P}$ of $G$, let $|\mc{P}|$ denote the number of paths in $\mc{P}$ and
we always have 
\begin{equation}
\label{eq01}
|V| = |\mc{P}| + |E(\mc{P})|.
\end{equation}
Therefore, to some extent, minimizing the number of paths is equivalent to maximizing the number of edges in the $k$-pp.
This induces the following problem, denoted by {\em $k$PPE} for simplicity, 
where ``E'' indicates that the objective is to cover as many edges as possible.

\begin{problem}
\label{pro02}
{\em ($k$PPE)}
Given $G$, 
the problem aims to find a $k$-pp of $G$ such that the {\em total number of edges in the $k$-pp} is maximized.
\end{problem}

\begin{lemma}
\label{lemma01}
An $\alpha$-approximation algorithm for $k$PPE is a $((1-\alpha)k+\alpha)$-approximation algorithm for $k$PP.
\end{lemma}
\begin{proof}
Suppose that we have an $\alpha$-approximation algorithm $A$ for $k$PPE and let $\mc{P}$ be the $k$-pp produced by $A$.
Let $\mc{Q}$ be a minimum $k$PP for $G$ and thus $|\mc{Q}| \ge \frac {|V|}k$. 
Since $\mc{Q}$ is a feasible solution of $k$PPE and $A$ is $\alpha$-approximation, we have $|E(\mc{P})| \ge \alpha \times |E(\mc{Q})|$.
It follows from Eq.(\ref{eq01}) that 
\[
|\mc{P}| = |V|-|E(\mc{P})| \le |V|-\alpha |E(\mc{Q})| = \alpha |\mc{Q}| + (1-\alpha) |V|
\le (\alpha+(1-\alpha)k) |\mc{Q}|,
\]
which indicates that $A$ is indeed a $((1-\alpha)k+\alpha)$-approximation algorithm for $k$PP.
The proof is completed.
\end{proof}

By Lemma~\ref{lemma01}, in the following, we consider $k$PPE instead.
For ease of presentation, we hereafter fix an optimal $k$-pp $\mc{Q}$ for $k$PPE for discussion.

\section{Initial steps}

We will present a $\frac 45$-approximation algorithm and an improved $\frac {9-\sqrt{11}}7$-approximation for $k$PPE with 
$k\in\{9,10\}$ and $k \ge 11$ in Sections~\ref{9PP} and~\ref{11PP}, respectively.
By Lemma~\ref{lemma01}, this leads to a $\frac {k+4}5$-approximation algorithm and 
a $(\frac {\sqrt{11}-2}7 k + \frac {9-\sqrt{11}}7)$-approximation algorithm for $k$PP with $k\in\{9,10\}$ and $k \ge 11$, respectively.
The two algorithms share several common initial steps.
The first step is to compute a maximum triangle-free path-cycle cover $\mc{F}$ in the graph $G$ in $O(n^3 m^2)$ time~\cite{Har24}.
Note that $\mc{Q}$ is a triangle-free path-cycle cover and thus we have $|E(\mc{F})| \ge |E(\mc{Q})|$.
But $\mc{F}$ may not be feasible for $k$PPE because of the existence of cycles and $(k+1)^+$-paths.

We next prove two lemmas for the sake of analyzing our algorithms later.

\begin{lemma}
\label{lemma02}
If $\ell$ is an integer larger than $k$, then $\frac {\ell - \lceil \ell/k \rceil}\ell \ge \frac {k-1}{k+1}$.
\end{lemma}
\begin{proof}
If $\ell$ is a multiple of $k$, then we have $\frac {\ell - \lceil \ell/k \rceil}\ell = \frac {k-1}k > \frac {k-1}{k+1}$.
We next assume $\ell = ks+t$ where $s \ge 1$ and $1 \le t <k$.
It follows that 
\[
\frac {\ell - \lceil \ell/k \rceil}\ell = 1-\frac {s+1}{ks+t} \ge 1-\frac {s+1}{ks+1} = \frac {k-1}{k+1/s} \ge \frac {k-1}{k+1}.
\]
The proof is finished.
\end{proof}

\begin{lemma}
\label{lemma03}
Suppose that $C$ is an $\ell$-path or $\ell$-cycle.
The following two statements hold:
\begin{itemize}
\item[1.] If $k \in \{ 9, 10 \}$ and $\ell \ge 5$, then we can construct a $k$-pp of $C$ containing $\frac 45 \ell$ edges.

\item[2.] If $k \ge 11$ and $\ell \ge 6$, then we can construct a $k$-pp of $C$ containing $\frac 56 \ell$ edges.
\end{itemize}
\end{lemma}
\begin{proof}
Note that if $C$ is an $\ell$-cycle, then it becomes an $\ell$-path after removing one edge in $C$.
Therefore, we assume $C$ is an $\ell$-path.

If $\ell \le k$, then $C$ itself is a $k$-pp with $\ell-1$ edges and we are done.
We next assume that $\ell \ge k+1$.
Clearly, the path $C$ can be cut into $\lfloor \ell/k \rfloor$ $k$-paths and an $(\ell-k \lfloor \ell/k \rfloor)$-path, 
which together form a $k$-pp $\mc{P}$ of $C$.
If $\ell$ is a multiple of $k$, then $\mc{P}$ contains $\ell/k\cdot (k-1) = \ell - \ell/k$ edges.
Otherwise, the number of edges in $\mc{P}$ is $\lfloor \ell/k \rfloor (k-1) + \ell - k\lfloor \ell/k \rfloor - 1 = \ell - \lceil \ell/k \rceil$.
By Lemma~\ref{lemma02}, $\mc{P}$ has at least $\frac {k-1}{k+1} \ell$ edges.
The lemma is proved since when $k \in \{ 9, 10 \}$, $\frac {k-1}{k+1} \ge \frac 45$ and when $k \ge 11$, $\frac {k-1}{k+1} \ge \frac 56$.
\end{proof}

Recall that $\mc{F}$ contains no triangles. In other words, every cycle in $\mc{F}$ has order at least $4$.
We say a cycle in $\mc{F}$ is {\em short} if its order is exactly $4$ or $5$.
If $\mc{F}$ has no $4$-cycles and $k \in \{ 9, 10 \}$, then by Statement~1 in Lemma~\ref{lemma03}, 
we can construct a $k$-pp for $G$ containing $\frac 45 |E(\mc{F})|$ edges.
One sees that this achieves an approximation ratio of $\frac 45$ because $|E(\mc{F})| \ge |E(\mc{Q})|$.
Similarly, if $\mc{F}$ has no short cycles and $k \ge 11$, then by Statement~2 in Lemma~\ref{lemma03}, 
we can easily achieve an approximation ratio of $\frac 56$.

\subsection{Connecting short cycles in $\mc{F}$ to the outside}
By the discussion in the last paragraph, we have to present a process to deal with short cycles in $\mc{F}$.
In fact, to connect as many short cycles as possible, our idea is to compute another path-cycle cover $\mc{W}$ in an auxiliary graph $G'$
so that more edges can be involved in the final solution.
The following definition is for constructing the graph $G'$.

\begin{definition}
\label{def01}
The auxiliary graph $G'$ is a spanning subgraph of $G$
where its edge set $E(G')$ consists of all edges $\{u, v\} \in E(G)$ 
such that $u$ and $v$ are in different connected components of $\mc{F}$ and at least one of them is in a short cycle.
\end{definition}

A short cycle $C$ in $\mc{F}$ is {\em saturated} by an edge $e \in E(G')$ if $e$ is incident to a vertex in $C$.
Similarly, a short cycle $C$ is {\em saturated} by a subset $E' \subseteq E(G')$, if $C$ is saturated by at least one edge in $E'$.

\begin{definition}
\label{def02}
Let $\eta = \mathbf{1}_{\{ k \ge 11 \}}$ where $\mathbf{1}$ is the indicator function.
Given a subset $E' \subseteq E(G')$, for $i \in \{ 4, 5 \}$, let $\mc{F}_i$ be the set of $i$-cycles in $\mc{F}$ saturated by $E'$.
The {\em weight} of $E'$ is 
\begin{equation}
\label{eq02}
|\mc{F}_4| + \eta \times |\mc{F}_5|.
\end{equation}
In other words, when $k \in \{ 9, 10 \}$, the weight is the number of $4$-cycles saturated by $E'$
and when $k \ge 11$, the weight becomes the number of short cycles saturated by $E'$.

A {\em maximum-weight} path-cycle cover of $G'$ is a path-cycle cover $\mc{W}$ of $G'$ such that 
the weight of $E(\mc{W})$ is maximized.
\end{definition}

\begin{lemma}
\label{lemma04}
A maximum-weight path-cycle cover $\mc{W}$ in $G'$ can be computed in $O(mn \log n)$ time.
\end{lemma}
\begin{proof}
We first introduce the {\em maximum-weight $[f, g]$-factor} problem defined on an edge-weighted graph $G_1$
where $f, g$ are two functions mapping each vertex $v \in V(G_1)$ to non-negative integers $f(v), g(v)$ with $f(v) \le g(v)$.
An {\em $[f,g]$-factor} of $G_1$ is a spanning graph of $G_1$ 
such that the degree of each vertex $v$ is in the interval $[f(v), g(v)]$.
The {\em weight} of an $[f, g]$-factor of $G_1$ is the total weight of edges in the factor. 
A maximum-weight $[f, g]$-factor of $G_1$ can be computed in $O(m_1n_1\log n_1)$ time~\cite{Gab83}, 
where $m_1 = |E(G_1)|$ and $n_1 = |V(G_1)|$. 

We prove the lemma by a reduction to the maximum-weight $[f, g]$-factor problem.
We let $C_1, C_2, \ldots, C_s$ be the short cycles in $\mc{F}$.
From the graph $G'$ in Definition~\ref{def01}, we construct an edge-weighted graph $G_1 = (V \cup X, E(G') \cup F_1 \cup F_2)$ as follows:
\begin{itemize}
\parskip=0pt
\item $X = \{x_i, y_i, z_i \mid 1 \le i \le s \}$. 

\item $F_1 = \{\{x_i, v\}, \{y_i, v\} \mid v \in V(C_i), 1 \le i \le s \}$ and 
$F_2 = \{\{x_i, z_i\}, \{y_i, z_i\} \mid 1 \le i \le s \}$.

\item The weight of each edge in $E(G') \cup F_1$ is $0$.

\item The weights of edges in $F_2$ are defined as follows.
If $C_i$ is a $4$-cycle, then the weights of $\{ x_i, z_i \}, \{ y_i, z_i \}$ are both $1$.
Otherwise, $C_i$ is a $5$-cycle and the weights of $\{ x_i, z_i \}, \{ y_i, z_i \}$ are both $\eta$.

\item For each vertex $v$ in a short cycle, $f(v) = g(v) = 2$. 

\item For each vertex $v$ in $V$ but not in a short cycle, $f(v) = 0$ and $g(v) = 2$.

\item For each $i \in\{ 1, 2, \ldots, s \}$, $f(x_i) = f(y_i) = f(z_i) = 0$ and 
$g(x_i) = g(y_i) = |V(C_i)|$, $g(z_i) = 1$.
\end{itemize}
Clearly, $|V(G_1)|\le 4|V|$ and $|E(G_1)|\le |E(G')|+4|V|$.
Therefore, a maximum-weighted $[f, g]$-factor $\mc{S}$ of $G_1$ can be computed in $O(mn \log n)$ time because $G$ is a connected graph and hence $|E| \ge |V|-1$.
Let $\mc{W} = (V, E(G') \cap E(\mc{S}))$ and one sees that $\mc{W}$ is a path-cycle cover of $G'$ since $g(v) \le 2$ for each vertex $v \in V(G')$.
It suffices to show that the weight of $E(\mc{W})$ is maximized.

We first prove that the weight of $E(\mc{W})$ is not less than that of $\mc{S}$. 
By the construction of $G_1$, only the edges $\{x_i, z_i\}$ and $\{y_i, z_i\}$ may have positive weights.
Since $g(z_i) = 1$, at most one of $\{x_i, z_i\}$ and $\{y_i, z_i\}$, say $\{ x_i, z_i \}$ is in $E(\mc{S})$.
By $g(x_i) = |V(C_i)|$, there exists a vertex $v \in V(C_i)$ such that the edge $\{v, x_i\}$ is not in $E(\mc{S})$. 
Recall that $f(v) = g(v) = 2$ and thus $\mc{W}$ contains an edge incident to $v$.
In other words, $C_i$ is saturated by $E(\mc{W})$ and we are done. 

We next prove that given a path-cycle cover $\mc{W}_1$ of $G'$, the weight of $\mc{S}$ is not less than that of $E(\mc{W}_1)$.
We can construct an $[f, g]$-factor $\mc{S}_1$ for $G_1$ from $\mc{W}_1$ as follows: 
Initially, we set $\mc{S}_1 = \mc{W}_1$. 
For each short cycle $C_i$ and each vertex $v \in V(C_i)$, 
we perform the following operations according to the degree of $v$ in $\mc{W}_1$. 
\begin{itemize}
\parskip=0pt
\item If the degree of $v$ in $\mc{W}_1$ is~$0$, then add the edges $\{v, x_i\}, \{v, y_i\}$ to $E(\mc{S}_1)$.

\item If the degree of $v$ in $\mc{W}_1$ is~$1$, then add the edge $\{v, y_i\}$ to $E(\mc{S}_1)$
and further add the edge $\{x_i, z_i\}$ to $E(\mc{S}_1)$ if it is not in $E(\mc{S}_1)$.

\item If the degree of $v$ in $\mc{W}_1$ is~$2$, then add the edge $\{ x_i, z_i \}$ to $E(\mc{S}_1)$ if it is not in $E(\mc{S}_1)$.
\end{itemize}
One sees that $\mc{S}_1$ is an $[f, g]$-factor of $G_1$ and thus the weight of $\mc{S}$ is not less than that of $\mc{S}_1$.
It remains to show that $\mc{S}_1$ and $E(\mc{W}_1)$ have the same weight.
To prove this, let $C_i$ be a short cycle saturated by $E(\mc{W}_1)$. 
Then, there exists a vertex $v \in V(C_i)$ such that $\mc{W}_1$ contains an edge incident to $v$. 
By the construction of $\mc{S}_1$, $\mc{S}_1$ contains $\{x_i, z_i\}$ and we are done.
The proof is completed since we have proved that the weight of $E(\mc{W})$ is at least as large as that of $\mc{S}$.
\end{proof}

Given a maximum-weight path-cycle cover $\mc{W}$ of $G'$, we say that $\mc{W}$ is {\em stingy}
if for every edge $e \in \mc{W}$, removing $e$ from $\mc{W}$ makes the weight of $E(\mc{W})$ smaller. 
As long as $\mc{W}$ is not stingy, we can keep removing some edge $e$ from $E(\mc{W})$ without changing the weight of $E(\mc{W})$. 
So, we can hereafter assume that $\mc{W}$ is stingy.
We conclude the initial four steps in Figure~\ref{fig01}.

\begin{figure}[htb]
\begin{center}
\framebox{
\begin{minipage}{5.1in}
Algorithm {\sc Approx1}:\\
Input: A graph $G = (V, E)$;
\begin{itemize}
\parskip=0pt
\item[1.]
	Compute a maximum triangle-free path-cycle cover $\mc{F}$ of $G$.
\item[2.]
         Construct the auxiliary graph $G'$ in Definition~\ref{def01} and compute a maximum-weight path-cycle cover $\mc{W}$ of $G'$ as in Lemma~\ref{lemma04}.
\item[3.]
	Obtain a stingy $\mc{W}$ by repeatedly removing an edge $e$ from $E(\mc{W})$ 
	as long as the removal does not change the weight of $E(\mc{W})$.
\item[4.]
     Output the graph $(V, E(\mc{F}) \cup E(\mc{W}))$.
\end{itemize}
\end{minipage}}
\end{center}
\caption{The initial steps of our algorithms.} \label{fig01}
\end{figure}

In the sequel, $\mc{W}$ always means a stingy maximum-weight path-cycle cover of $G'$.
Let $\mc{F}+\mc{W}$ be the graph $(V, E(\mc{F}) \cup E(\mc{W}))$. 
For $i \in \{ 4, 5 \}$, let $\mc{I}_i$ be the set of $i$-cycles in $\mc{F}$ not saturated by $E(\mc{W})$.
A connected component $K$ of $\mc{F}+\mc{W}$ is {\em critical} if it belongs to 
$\mc{I}_4\cup \mc{I}_5$; otherwise, it is {\em non-critical}.
The next lemma will be crucial for proving our desired approximation ratios.

\begin{lemma}
\label{lemma05}
$|E(\mc{Q})| \le |E(\mc{F})| - |\mc{I}_4| - \eta |\mc{I}_5|$.
\end{lemma}
\begin{proof}
For $i \in \{ 4, 5 \}$, let $s_i$ be the number of $i$-cycles in $\mc{F}$, and $\mc{J}_i$ be the set of $i$-cycles in $\mc{F}$ 
not saturated by the edge set $E(\mc{Q}) \cap E(G')$.
Note that the weight of $E(\mc{Q}) \cap E(G'))$ and $\mc{W}$ is $s_4-|\mc{J}_4| + \eta(s_5-|\mc{J}_5|)$ and $s_4-|\mc{I}_4| + \eta(s_5-|\mc{I}_5|)$, respectively.
Since $\mc{Q}$ is a $k$-pp, the graph $(V, E(\mc{Q}) \cap E(G'))$ is a path-cycle cover of $G'$.
Therefore, $s_4-|\mc{J}_4| + \eta(s_5-|\mc{J}_5|) \le s_4-|\mc{I}_4| + \eta(s_5-|\mc{I}_5|)$ and 
it follows that $|\mc{J}_4|  + \eta |\mc{J}_5|  \ge |\mc{I}_4| + \eta |\mc{I}_5|$.

We prove the lemma by discussing whether or not $\eta = 1$.

\medskip
{\em Case~1: } $\eta = 0$ and thus $|\mc{J}_4| \ge |\mc{I}_4|$.
Let $U = V \setminus V(\mc{J}_4)$. By the definition of $\mc{J}_4$,
we know that for each edge $\{ u, v \} \in E(\mc{Q})$, either $\{u, v\} \subseteq U$, or some cycle in $\mc{J}_4$ contains both $u$ and $v$.
Since $\mc{Q}$ is a path-cycle of $G$, each cycle in $\mc{J}_4$ has at most $3$ edges in $\mc{Q}$.
Therefore, $|E(\mc{Q})| \le |E(\mc{Q}[U])| + 3|\mc{J}_4|$.
Similarly, $E(\mc{F})$ can be partitioned into two subsets $E(\mc{F}[U])$ and $E(\mc{J}_4)$.
Note that $\mc{F}[U]$ is a maximum triangle-free path cycle cover of $G[U]$ and 
thus $|E(\mc{F}[U])| \ge |E(\mc{Q}[U])|$.
We finally obtain
\[
|E(\mc{Q})| \le |E(\mc{Q}[U])| + 3|\mc{J}_4| \le |E(\mc{F}[U])| + 3|\mc{J}_4| 
= |E(\mc{F})| - |\mc{J}_4| \le |E(\mc{F})| - |\mc{I}_4|.
\]

\medskip
{\em Case~2: } $\eta = 1$.
In this case, $|\mc{J}_4|  + |\mc{J}_5|  \ge |\mc{I}_4| + |\mc{I}_5|$.
The proof is very similar to Case~1 but we may need to set $U = V \setminus (V(\mc{J}_4) \cup V(\mc{J}_5))$.
The lemma is proved.
\end{proof}

\subsection{Possible structures of connected components in $\mc{F}+\mc{W}$}

In this subsection, we analyze the possible structures of a connected component $K$ in $\mc{F}+\mc{W}$.
We remind the readers that the maximum-weight path-cycle cover $\mc{W}$ is stingy.

\begin{definition}
\label{def03}
Given $K$, let $(K)_m$ denote the graph obtained from $K$ by contracting each connected component of $\mc{F}$ to a single node.
In other words, the nodes of $(K)_m$ one-to-one correspond to the connected components of $\mc{F}$
and two nodes are adjacent in $(K)_m$ if and only if $\mc{W}$ contains an edge between the two corresponding components.
\end{definition}

\begin{lemma}
\label{lemma06}
$(K)_m$ is a single node, edge or star and the following three statements hold:
\begin{enumerate}
\parskip=0pt
\item If $(K)_m$ is an edge, then at least one endpoint of $(K)_m$ corresponds to a short cycle in $\mc{F}$.

\item If $(K)_m$ is a star, then each satellite of $(K)_m$ corresponds to a short cycle in $\mc{F}$.

\item If $k \in \{ 9, 10 \}$, then ``short cycle'' in the above two statements can be replaced by ``$4$-cycle''.
\end{enumerate}
\end{lemma}
\begin{proof}
Clearly, if $(K)_m$ has no edges, then $(K)_m$ is a single node.
We next assume $(K)_m$ has at least one edge, which corresponds to an edge in $\mc{W}$.
Recall that $\mc{W}$ is stingy.
If there were a $4$-path or cycle in $(K)_m$, then removing a certain edge of $\mc{W}$ will not change the weight of 
$E(\mc{W})$, contradicting the stinginess of $\mc{W}$.
Therefore, $(K)_m$ contains neither a $4$-path nor a cycle.
In other words, $(K)_m$ is an edge or star.

We first assume $(K)_m = \{ u, v \}$ is an edge.
By the construction of $G'$, at least one of $u$ and $v$ corresponds to a short cycle of $\mc{F}$ and thus Statement~1 is proved.
When $k \in \{ 9, 10 \}$, $\eta = 0$; so, if neither $u$ nor $v$ corresponds to a $4$-cycle, 
then the edge in $\mc{W}$ corresponding to $\{ u, v \}$ can be deleted without changing the weight of $E(\mc{W})$, 
a contradiction against the stinginess of $\mc{W}$.
Therefore, Statement~3 holds when $(K)_m$ is an edge.

We next assume $(K)_m$ is a star centered at $u$.
Let $v$ be a satellite of $(K)_m$ and thus $\{ u, v \}$ is an edge in $(K)_m$.
If $v$ does not correspond to a short cycle in $\mc{F}$, then since $(K)_m$ is a star, 
removing the edge in $\mc{W}$ corresponding to $\{ u, v \}$ from $\mc{W}$ will not change the weight of $E(\mc{W})$, contradicting the stinginess of $\mc{W}$.
This proves Statement~2.
Similarly, when $k \in \{ 9, 10 \}$, $\eta = 0$ and $v$ must correspond to a $4$-cycle. 
This completes the proof.
\end{proof}

If $(K)_m$ is a single node or edge, then we select a node of $(K)_m$ as the {\em center} of $(K)_m$ as follows.
If $(K)_m$ is a single node, then it becomes the center.
If $(K)_m$ is an edge, then by Statements~1 and~3 in Lemma~\ref{lemma06}, $(K)_m$ has a node $v$ corresponding to 
a short cycle or more specifically a $4$-cycle for $k \in \{ 9, 10 \}$; 
one arbitrary such $v$ becomes the satellite, while the other node becomes the center.
In this way, even when $(K)_m$ is not a star, it has exactly one center and each satellite of $(K)_m$ corresponds to a short cycle in $\mc{F}$
or more specifically a $4$-cycle in $\mc{F}$ for $k \in \{ 9, 10 \}$.

\begin{definition}
\label{def04}
The {\em center element} (respectively, a {\em satellite element}) of $K$ is the connected component in $\mc{F}$ corresponding to 
the center (respectively, a satellite) of $(K)_m$.
Let $K_c$ denote the center element and call each vertex in $K_c$ an {\em anchor}.

A vertex in $K_c$ is a {\em $j$-anchor} if it is incident to exactly $j$ edges in $\mc{W}$.
Recall that $\mc{W}$ is a path-cycle cover and thus $j \in \{ 0, 1, 2 \}$. 

For a satellite element $C$ of a connected component of $\mc{F}+\mc{W}$, 
we use $\varepsilon(C)$ to denote the unique edge of $\mc{W}$ incident to a vertex of $C$, 
and use $\nu(C)$ to denote the endpoint of $\varepsilon(C)$ contained in $C$.
\end{definition}

Since $K_c$ is a path or $4^+$-cycle, we next define several notations for convenience in later discussion.

\begin{notation}
\label{nota01} 
Let $\ell$ be the order of $K_c$.
The vertices in $K_c$ can be ordered as a sequence $v_1, v_2, \ldots v_\ell$ in such a way 
that $\{v_i, v_{i+1}\}$ is an edge in $K_c$ for every $1\le i<\ell$.
If $K_c$ is a cycle, then $\{ v_\ell, v_1 \}$ is also an edge in $K_c$ and 
we define $v_{\ell+1} = v_1$, $v_0 = v_\ell$ for simplicity.

For each vertex $v_i \in V(K_c)$ that is a $1$-anchor, we use $C_i$ to denote the short cycle in $\mc{F}$ 
with $\varepsilon(C_i) = \{v_i, \nu(C_i)\}$.
Similarly, for each vertex $v_i \in V(K_c)$ that is a $2$-anchor, we use $C_i$ and $C'_i$ to denote the two 
short cycles in $\mc{F}$ with 
$\varepsilon(C_i) = \{v_i, \nu(C_i)\}$ and $\varepsilon(C'_i) = \{v_i, \nu(C'_i)\}$.

For two integers $i, j$ with $i \le j$, 
we use $K[i, j]$ to denote the graph obtained from $K$ by 
\begin{itemize} 
\item first removing the edges $\{ v_{i-1}, v_i \}$ and $\{ v_j, v_{j+1} \}$, if they are present; 
\item and then removing all vertices outside the connected component containing $v_i$ and $v_j$.
\end{itemize}
When $i = j$, we simplify $K[i, j]$ to $K[i]$. See Figure~\ref{fig02} for an illustration. 
We call the subgraph of $K_c$ still remaining in $K[i,j]$ the {\em center element} of $K[i,j]$, 
and call each sattelite element of $K$ still remaining in $K[i,j]$ a {\em satellite element} of $K[i,j]$.
Clearly, the center element of $K[i, j]$ is a path even if $K_c$ is a cycle. 
As in Definition~\ref{def03}, we define $(K[i,j])_m$ to be the graph obtained from $K[i,j]$ 
by contracting each of its center element and satellite elements to a single node.

\end{notation}

\input{fig02}

\section{A $\frac 45$-Approximation Algorithm for $k$PPE with $k \in \{ 9, 10 \}$}
\label{9PP}

In this section, we present a $\frac 45$-approximation algorithm for $k$PPE with $k \in \{ 9, 10 \}$,
which is indeed a $\frac {k+4}5$-approximation algorithm for $k$PP by Lemma~\ref{lemma01}.
We first call {\sc Approx1} in Figure~\ref{fig01} to obtain the spanning graph $\mc{F}+\mc{W}$.
Recall that $\mc{I}_4$ is the set of critical $4$-cycles in $\mc{F}+\mc{W}$.
Let $\mc{T}$ and $\mc{O}$ be the set of $2$-anchors and $1$-anchors in $\mc{F}+\mc{W}$, respectively. 
For a subgraph $H$ of $\mc{F}+\mc{W}$, we define $F_H = E(H) \cap E(\mc{F})$; 
in other words, $F_H$ consists of all edges in both $H$ and $\mc{F}$.
Moreover, 
we define $\mc{T}_H = \mc{T} \cap V(H)$ and $\mc{O}_H = \mc{O} \cap V(H)$; 
in other words, $\mc{T}_H$ (respectively, $\mc{O}_H$) consists of 2-anchors (respectively, 1-anchors) 
of $\mc{F}+\mc{W}$ appearing in $H$.

For the discussions in Lemmas~\ref{lemma07}--\ref{lemma09}, 
let $K$ be a connected component of $\mc{F}+\mc{W}$ not contained in $\mc{I}_4$.
Since $K \notin \mc{I}_4$, Lemma~\ref{lemma06} implies that (1)~$(K)_m$ is a single node but $K$ is not a $4$-cycle, 
(2)~$(K)_m$ is not a single node and $\mc{T}_K = \emptyset$, or (3)~$(K)_m$ is not a single node and $\mc{T}_K \ne \emptyset$.
We next prove that in any case, $K$ has a $k$-pp with at least $\frac 45 |F_K|$ edges. 
For this end, we use Notation~\ref{nota01}.

\begin{lemma}
\label{lemma07}
Suppose that $(K)_m$ is a single node and $K \notin \mc{I}_4$. 
Then, we can construct a $k$-pp of $K$ with at least $\frac 45 |F_K|$ edges.
\end{lemma}
\begin{proof}
Note that $K$ is a path or $5^+$-cycle in $\mc{F}$.
Therefore, $F_K = E(K)$.
If $K$ is a $k^-$-path, then $K$ itself is a $k$-pp and we are done.
Otherwise, $K$ is a $(k+1)^+$-path or $5^+$-cycle.
By Statement~1 in Lemma~\ref{lemma03}, the proof is completed.
\end{proof}

\begin{lemma}
\label{lemma08}
Suppose that $(K)_m$ is not a single node and $\mc{T}_K = \emptyset$.
Then, we can construct a $k$-pp of $K$ with at least $\frac 45 |F_K|$ edges.
\end{lemma}
\begin{proof}
To construct a $k$-pp of $K$, we first construct a graph $K'$ from $K$ by removing 
all edges $\{ v_{i-1}, v_i \}$ of $K$ such that $1\le i \le \ell$ and $v_i \in \mc{O}_K$. 
Let $\Gamma$ be the set of connected components of $K'$.
It suffices to show that $K'$ contains a $k$-pp with at least $\frac 45 |F_K|$ edges.

\medskip
{\em Case~1. } $K_c$ is a cycle.
In this case, $K$ contains the edge $\{ v_{i-1}, v_i \}$ for every $v_i \in \mc{O}_K$.
Therefore, $|F_{K'}| = |F_K| - |\mc{O}_K|$ and $|\Gamma| = |\mc{O}_K|$.
Each $H \in \Gamma$ is clearly $K[i, j]$ for some $1 \le i \le j \le \ell$ and 
contains a unique element of $\mc{O}_K$ (namely, $v_i$). 
See Figure~\ref{fig03} for an illustration. 
Obviously, $|F_H| = j-i+4$ and we can turn $H$ into a $(j-i+5)$-path by removing one edge in $C_i$ incident to 
$\nu(C_i)$.
Therefore, by Statement~1 in Lemma~\ref{lemma03} and $j-i+5\ge 5$, 
$H$ contains a $k$-pp with $\frac 45 (j-i+5) = \frac 45 |F_H| + \frac 45$ edges.
By combining these $k$-pps of all $H \in \Gamma$, we obtain a $k$-pp of $K'$
with $\sum_{H \in \Gamma} (\frac 45 |F_H| + \frac 45)$ edges.
Recall that $|F_{K'}| = |F_K| - |\mc{O}_K|$, $|\Gamma| = |\mc{O}_K|$ and $|F_{K'}| = \sum_{H \in \Gamma} |F_H|$.
We now have
\[
\sum_{H \in \Gamma} \left(\frac 45 |F_H| + \frac 45\right) = \frac 45 |F_{K'}| + \frac 45 |\Gamma| 
= \frac 45 |F_{K'}| + \frac 45 |\mc{O}_K| = \frac 45 |F_K|.
\]
The proof is finished in this case.

\medskip
{\em Case~2. } $K_c$ is a path and $v_1$ is a $1$-anchor. 
In this case, we removed $|\mc{O}_K|-1$ edges from $K_c$ in the construction of $K'$ 
because $\{ v_{i-1}, v_i \}$ is not an edge in $K_c$ when $i = 1$.
It follows that $|F_{K'}| = |F_K| - |\mc{O}_K| + 1$ and $|\Gamma| = |\mc{O}_K|$.
Note that each $H \in \Gamma$  is $K[i, j]$ for some $1 \le i \le j \le \ell$ and contains a unique element $v_i$ in $\mc{O}_K$.
Similarly to Case~1, we can construct a $k$-pp of $H$ with $\frac 45 |F_H| + \frac 45$ edges.
Therefore, $K'$ contains a $k$-pp with at least 
\[
\sum_{H \in \Gamma} \left(\frac 45 |F_H| + \frac 45\right) = \frac 45 |F_{K'}| + \frac 45 |\Gamma| 
= \frac 45 |F_{K'}| + \frac 45 |\mc{O}_K| = \frac 45 |F_K| + \frac 45 > \frac 45 |F_K|
\]
edges.
The proof is finished in this case.

\medskip
{\em Case~3. } $K_c$ is a path and $v_1$ is not a $1$-anchor. 
In this case, $\{ v_{i-1}, v_i \}$ is an edge in $K_c$ for every $v_i \in \mc{O}_K$.
Therefore, $|F_{K'}| = |F_K| - |\mc{O}_K|$ and $|\Gamma| = |\mc{O}_K|+1$.
Let $H_1 \in \Gamma$ be the connected component of $\Gamma$ containing $v_1$.
Clearly, $H_1$ contains no $1$-anchors of $K$ and hence is a path.
By Lemma~\ref{lemma07}, $H_1$ has a $k$-pp with $\frac 45 |E(H_1)| = \frac 45 |F_{H_1}|$ edges.
Each $H \in \Gamma \setminus \{H_1\}$  is $K[i, j]$ for some $1 \le i \le j \le \ell$ and 
contains a unique $v_i \in \mc{O}_K$.
Similarly to Case~1, $H$ has a $k$-pp with $\frac 45 |F_H| + \frac 45$ edges.
Therefore, we can construct a $k$-pp of $K'$ with at least 
$\frac 45 |F_{H_1}| + \sum_{H \in \Gamma \setminus \{H_1\}} \left(\frac 45 |F_H| + \frac 45\right)$ edges.
The proof in this case is finished because 
\[
\frac 45 |F_{H_1}| + \sum_{H \in \Gamma \setminus \{H_1\}}\left (\frac 45 |F_H| + \frac 45\right) 
= \frac 45 |F_{K'}| + \frac 45 ( |\Gamma|-1 ) 
= \frac 45 |F_{K'}| + \frac 45 |\mc{O}_K| = \frac 45 |F_K|.
\]

Since Case~1, 2, or~3 must occur, the lemma is proved.
\end{proof}


\begin{lemma}
\label{lemma09}
Suppose that $(K)_m$ is not a single node and $\mc{T}_K \ne \emptyset$.
Then, we can construct a $k$-pp of $K$ with at least $\frac 45 |F_K|$ edges.
\end{lemma}
\begin{proof}
To construct a $k$-pp of $K$, we first construct a graph $K'$ from $K$ by first 
removing the edges of $K_c$ incident to $v_i$ for every $v_i \in \mc{T}_K$. 
Let $\Gamma$ be the set of connected components of $K'$, and
$\Gamma_1$ be the set of connected components of $K'$ containing at least one vertex of $\mc{T}_K$.
Since we remove at most $2|\mc{T}_K|$ edges from $K$ to construct $K'$, we have $|F_{K'}| \ge |F_K| - 2|\mc{T}_K|$.

By the construction of $K'$, each $H \in \Gamma_1$ is $K[i]$ for some $v_i \in \mc{T}_K$, and hence $|F_H| = 8$. 
For an illustration, the subgraph $K[1]$ in the left graph in Figure~\ref{fig02} shows an example of $H$.
Moreover, $|\Gamma_1| = |\mc{T}_K|$.
Clearly, we can easily construct a $9$-path from each $H\in \Gamma_1$.
In other words, each $H\in \Gamma_1$ has a $k$-pp with $8 = \frac 45 |F_H| + \frac 85$ edges.
On the other hand, each $H \in \Gamma \setminus \Gamma_1$ is $K[i, j]$ for some $1 \le i \le j \le \ell$
such that $(H)_m$ is either a single node, or $(H)_m$ is not a single node but $\mc{T}_H = \emptyset$.
Although $H$ may not be a connected component of $\mc{F}+\mc{W}$, 
we can construct a $k$-pp for $H$ with at least $\frac 45 |F_H|$ edges using the same construction 
in Lemmas~\ref{lemma07} and~\ref{lemma08}.
By combining these $k$-pps of all $H\in\Gamma$, we obtain a $k$-pp of $K'$ with 
$\sum_{H \in \Gamma_1} \left(\frac 45 |F_H| + \frac 85\right) + \sum_{H \in \Gamma \setminus \Gamma_1} \frac 45 |F_H|$ edges.
Recall that $|F_{K'}| \ge |F_K| - 2|\mc{T}_K|$ and $|\Gamma_1| = |\mc{T}_K|$.
We have
\[
\sum_{H \in \Gamma_1}\left(\frac 45 |F_H| + \frac 85\right) + \sum_{H \in \Gamma \setminus \Gamma_1} \frac 45 |F_H|
= \frac 45 |F_{K'}| + \frac 85 |\Gamma_1| = \frac 45 |F_{K'}| + \frac 85 |\mc{T}_K| \ge \frac 45 |F_{K}|.
\]

The lemma is proved.
\end{proof}

\input{fig03}

\begin{theorem}
\label{thm01}
For each $k \in \{ 9, 10 \}$, there is an $O(n^3m^2)$-time $\frac 45$-approximation algorithm for $k$PPE.
\end{theorem}
\begin{proof}
The algorithm proceeds as follows. 
First, it calls {\sc Approx1} to obtain the graph $\mc{F}+\mc{W}$.
Let $K$ be a connected component of $\mc{F}+\mc{W}$.
If $K \in \mc{I}_4$, the algorithm outputs a $4$-path of $K$. 
Otherwise, it proceeds as in the proofs of Lemmas~\ref{lemma07}, \ref{lemma08}, and~\ref{lemma09} 
to output a $k$-pp of $K$ with at least $\frac 45 |F_K|$ edges. 
The outputted paths (over all connected components $K$) together form a $k$-pp of $G$.

We next analyze the approximation ratio.
The number of edges in $\mc{F}$ but not in the 4-cycles in $\mc{I}_4$ is exactly $|E(\mc{F})|-4|\mc{I}_4|$.
Recall that $\eta = 0$.
By Lemmas~\ref{lemma05}, \ref{lemma07}, \ref{lemma08}, and~\ref{lemma09}, 
the number of edges in the final combined $k$-pp is at least 
\[
\frac 45 (|E(\mc{F})|-4|\mc{I}_4|) + 3|\mc{I}_4| = \frac 45 |E(\mc{F})| - \frac 15 |\mc{I}_4|
\ge \frac 45 |E(\mc{Q})| + \frac 35 |\mc{I}_4| \ge \frac 45 |E(\mc{Q})|.
\]

The time complexity of the algorithm is dominated by the computation of $\mc{F}$, 
because it requires $O(n^3 m^2)$ time but the other steps clearly need less time.
The proof is completed.
\end{proof}

\section{A $\frac {9-\sqrt{11}}7$-Approximation Algorithm for $k$PPE when $k \ge 11$}
\label{11PP}

In this section, we present an improved $\frac {9-\sqrt{11}}7$-approximation algorithm for $k$PPE when $k \ge 11$.
Let $r = \frac {9-\sqrt{11}}7 \approx 0.8119$ be the smaller root of the equation $7x^2-18x+10=0$.
Similarly to Section~\ref{9PP}, for a subgraph $H$ of $\mc{F}+\mc{W}$, 
we define the notations $F_H = E(H) \cap E(\mc{F})$, $\mc{T}_H = \mc{T} \cap V(H)$ and $\mc{O}_H = \mc{O} \cap V(H)$.
The next subsection discusses the case of non-critical connected components $K$ of $\mc{F}+\mc{W}$ 
such that $K_c$ has no $2$-anchors.

\subsection{Non-critical connected components of $\mc{F}+\mc{W}$ with no 2-anchors}

Throughout this subsection, $K$ means a non-critical connected component of $\mc{F}+\mc{W}$ such that $\mc{T}_K = \emptyset$.

\begin{lemma}
\label{lemma10}
Suppose that no satellite element of $K$ is a $4$-cycle. 
Then, we can construct a $k$-pp of $K$ with
\begin{itemize}
\item[1.] at least $r |F_K| + (5-6r) |\mc{O}_K|$ edges, or

\item[2.] at least $r |F_K| + r + (5-6r) |\mc{O}_K|$ edges if $K_c$ is a path and $v_1$ is a $1$-anchor.
\end{itemize}
\end{lemma}
\begin{proof}
We first assume $K$ has no satellite elements, i.e., $(K)_m$ is a single node.
It suffices to show that $K$ has a $k$-pp with at least $r |F_K|$ edges.
Since $K$ is non-critical, $K$ is a path or $6^+$-cycle and thus $F_K = E(K)$.
If $K$ has at least $6$ vertices, then the lemma follows from Statement~2 in Lemma~\ref{lemma03} and 
the inequality $r < \frac 56$.
Otherwise, $K$ is a $5^-$-path which alone forms a $k$-pp of $K$. 
So, in any case, we are done.

We next assume $(K)_m$ is not a single node. 
Let $K'$ be the graph obtained from $K$ by 
removing all edges $\{ v_{i-1}, v_i \}$ of $K$ such that $1\le i \le \ell$ and $v_i \in \mc{O}_K$. 
Let $\Gamma$ be the set of connected components of $K'$. 
We finish the proof by distinguishing three cases as follows.

\medskip
{\em Case~1.} $K_c$ is a cycle.
In this case, the edges $\{ v_{i-1}, v_i \}$ exist for every $v_i \in \mc{O}_K$. Thus, $|F_{K'}| = |F_K| - |\mc{O}_K|$ and $|\Gamma| = |\mc{O}_K|$.
Obviously, each $H \in \Gamma$ is $K[i, j]$ for some $1\le i\le j\le\ell$ and $v_i$ is the unique element in $\mc{O}_K$.
Since no satellite element of $K$ is a $4$-cycle, $C_i$ is a $5$-cycle.
It follows that $|F_H| = j-i+5$ and 
we can turn $H$ into a $(j-i+6)$-path by removing one edge of $C_i$ incident to $\nu(C_i)$.
By Statement~2 in Lemma~\ref{lemma03} and the inequality $j-i+6\ge 6$, $H$ has a $k$-pp with \[
\frac 56 (j-i+6) = r(j-i+5)+r+\left(\frac 56-r\right)(j-i+6) \ge r |F_H| + 5-5r.
\]
edges.
Therefore, we can combine these $k$-pps of all $H \in \Gamma$ to form a $k$-pp of $K'$ 
with $\sum_{H \in \Gamma} (r |F_H| + 5-5r)$ edges.
Now, since $|F_{K'}| = |F_K| - |\mc{O}_K|$, $|\Gamma| = |\mc{O}_K|$, and $|F_{K'}| = \sum_{H \in \Gamma} |F_H|$, 
we have
\[
\sum_{H \in \Gamma} (r |F_H| + 5-5r) = r |F_{K'}| + (5-5r) |\Gamma| 
= r |F_{K'}| + (5-5r) |\mc{O}_K| = r |F_K| + (5-6r) |\mc{O}_K|.
\]
The proof is done in this case.

\medskip
{\em Case~2.} $K_c$ is a path and $v_1$ is not a $1$-anchor.
As in Case~1, $\{ v_{i-1}, v_i \}$ is an edge of $K_c$ for every $v_i \in \mc{O}_K$.
Therefore, $|F_{K'}| = |F_K| - |\mc{O}_K|$ and $|\Gamma| = |\mc{O}_K|+1$.
Let $H_1 \in \Gamma$ be the connected component of $K'$ containing $v_1$. 
Clearly, $H_1$ is a path, and hence we can show that $H_1$ has a $k$-pp with $r |F_{H_1}|$ edges 
as in the first paragraph of the proof. 
Obviously, each $H \in \Gamma$ other than $H_1$ is $K[i, j]$ for some $1\le i\le j\le\ell$ and 
contains a unique $v_i \in \mc{O}_K$.
As in Case~1, we can show that $H$ has a $k$-pp with $r |F_H|+5-5r$ edges.
By combining these $k$-pps of all $H \in \Gamma$, we obtain a $k$-pp of $K$
with $r |F_{H_1}| + \sum_{H \in \Gamma \setminus \{ H_1 \}} (r |F_H| + 5-5r)$ edges.
Since $|F_{K'}| = |F_K| - |\mc{O}_K|$, $|\Gamma| = |\mc{O}_K|+1$, and $|F_{K'}| = \sum_{H \in \Gamma} |F_H|$, we have 
\begin{eqnarray*}
&   & r |F_{H_1}| + \sum_{H \in \Gamma \setminus \{ H_1 \}} (r |F_H| + 5-5r) \\
& = & r |F_{K'}| + (5-5r)(|\Gamma|-1) \\
& = & r |F_{K'}| + (5-5r)|\mc{O}_K| = r |F_K| + (5-6r)|\mc{O}_K|.
\end{eqnarray*}
The proof is finished in this case.

\medskip
{\em Case~3.} $K_c$ is a path and $v_1$ is a $1$-anchor. 
In this case, $\{ v_{i-1}, v_i \}$ is not an edge of $K_c$ when $i = 1$. 
Thus, $|F_{K'}| = |F_K| - (|\mc{O}_K|-1)$ and $|\Gamma| = |\mc{O}_K|$. 
Moreover, each $H \in \Gamma$ is $K[i,j]$ for some $1\le i\le j\le \ell$ and $v_i$ is the unique $1$-anchor in $\mc{O}_K$.
Similarly to Case~1, we can observe that $H$ contains a $k$-pp with $r |F_H| + 5-5r$ edges.
It follows that $K$ has a $k$-pp with $\sum_{H \in \Gamma} (r |F_H| + 5-5r)$ edges.
Since $|F_{K'}| = |F_K| - (|\mc{O}_K|-1)$, $|\Gamma| = |\mc{O}_K|$, and $|F_{K'}| = \sum_{H \in \Gamma} |F_H|$, we now have
\[
\sum_{H \in \Gamma} (r |F_H| + 5-5r) = r |F_{K'}| + (5-5r) |\Gamma|  = r |F_{K'}| + (5-5r)|\mc{O}_K| = r |F_K| + r + (5-6r)|\mc{O}_K|.
\]
The proof is completed in this case.
\end{proof}

Recall that a satellite element of $K$ is a $4$- or $5$-cycle by Statements~1 and~2 in Lemma~\ref{lemma06}.
Further recall that $K_c$ is an $\ell$-path or $\ell$-cycle.
The next definition introduces four {\em special} structures of $K$. 

\begin{definition}
\label{def05}
$K$ is {\em special} if $v_1$ is a $1$-anchor, $C_1$ is a $4$-cycle, 
and $K_c$ is an $\ell$-path satisfying one of Conditions~S1 through~S4 below (see Figure~\ref{fig04}):

\begin{itemize}
\item[S1.] $\ell=2$, $v_2$ is a $1$-anchor, and $C_2$ is a $5$-cycle.

\item[S2.] $\ell=3$, $v_2$ is a $0$-anchor, $v_3$ is a $1$-anchor, and $C_3$ is a $4$-cycle.

\item[S3.] $\ell=4$, $v_3$ is a $0$-anchor, both $v_2$ and $v_4$ are $1$-anchors, 
$C_2$ is a $4$-cycle, and $C_4$ is a $5$-cycle.

\item[S4.] There exists an integer $i \in \{1,2, \ldots, \ell\}$ such that 
$v_1 \ldots v_i$ are $1$-anchors, $v_{i+1}, \ldots, v_\ell$ are $0$-anchors, and $C_1, \ldots C_i$ are $4$-cycles.
\end{itemize}

For $j \in \{ 1, 2, 3, 4 \}$, we say that $K$ is {\em of type-$j$} if it is special and Condition~$Sj$ holds.
\end{definition}

\input{fig04}

\begin{lemma}
\label{lemma11}
We can construct a $k$-pp of $K$ with at least $r |F_K| + 2r + (5-6r) |\mc{O}_K|$ edges
if $K$ satisfies one of the following conditions:
\begin{itemize}
\item[1.] $K$ is of type-$1$, $2$, or $3$. 
(See the top row in Figure~\ref{fig04}.)

\item[2.] $K$ is of type-$4$ and $K_c$ is a $4$-path all of whose vertices are $1$-anchors.
(See the rightmost graph in the bottom row in Figure~\ref{fig04}.)
\end{itemize}
\end{lemma}
\begin{proof}
We here only consider the case where $K$ is of type-1; the other cases are similar.
Suppose $K$ is of type-$1$. Then, $|F_K| = 10$ and $|\mc{O}_K| = 2$.
Clearly, we can turn $K$ into a $11$-path by removing one edge incident to $\nu(C_i)$ in $C_i$ for each $i \in \{ 1, 2 \}$. 
The path alone forms a $k$-pp of $K$ because $k \ge 11$. 
The number of edges in the path is $10 = r |F_K| + 2r + (5-6r) |\mc{O}_K|$.
\end{proof}

\begin{lemma}
\label{lemma12}
Suppose that $K$ is of type-$4$. 
Then, we can construct a $k$-pp of $K$ with at least $r |F_K| + r + (5-6r) |\mc{O}_K|$ edges
if one of the following conditions holds.
\begin{itemize}
\item[1.] $K_c$ is a $2^+$-path and $v_1$ is the unique $1$-anchor in $K_c$.
(See the leftmost graph in the bottom row in Figure~\ref{fig04}.)

\item[2.] $K_c$ is not a $3$-path and the 1-anchors in $K_c$ are $v_1$ and $v_2$ only.
(See the middle graph in the bottom row in Figure~\ref{fig04}.)
\end{itemize}
\end{lemma}
\begin{proof}
Recall that $\ell = |V(K_c)|$. 
We first assume that Condition~1 holds. 
Then, $\ell \ge 2$, $|F_K| = \ell+3$, and $|\mc{O}_K| = 1$.
Therefore, $5-6r \le (\frac 56-r)(\ell+4)$ and it remains to show that $K$ has a $k$-pp $\mc{P}$ with at least $\frac 56 (\ell+4)$ edges. 
To obtain $\mc{P}$, we first turn $K$ into an $(\ell+4)$-path $P$ by removing one edge incident to $\nu(C_1)$ in $C_1$. 
Because $\ell+4 \ge 6$, Statement~2 in Lemma~\ref{lemma03} ensures that $\mc{P}$ can be constructed from $P$.

We next assume that Condition~2 holds. 
Then, $|F_K| = \ell+7$, $|\mc{O}_K| = 2$, and either $\ell = 2$ or $\ell \ge 4$.
If $\ell = 2$, then $|F_K| = 9$, $|\mc{O}_K| = 2$, and we can turn $K$ into a $10$-path which alone forms a $k$-pp of $K$ with $9$ edges; 
consequently, we are done because $r \ge \frac 12$.
So, we may assume $\ell \ge 4$.
Then, $10-12r \le (\frac 56-r)(\ell+8)$ and it remains to show that $K$ has a $k$-pp $\mc{P}$ with at least $\frac 56 (\ell+8)$ edges.
To obtain $\mc{P}$, first note that $K[1, 2]$ contians a $10$-path $P$ and $K[3, \ell]$ is an $(\ell-2)$-path.
If $\ell-2 \le k$, then $P$ and $K[3, \ell]$ together form a $k$-pp of $K$ with $\ell+6\ge \frac 56 (\ell+8)$ edges.
Otherwise, by Statement~2 in Lemma~\ref{lemma03}, $K[3, \ell]$ has a $k$-pp with $\frac 56 (\ell-2)$ edges; 
this $k$-pp and $P$ together form a $k$-pp of $K$ with at least 
$\frac 56 (\ell-2)+9 \ge \frac 56 (\ell+8)$ edges.
This finishes the proof.
\end{proof}

We are now ready to prove two main lemmas in this subsection.

\begin{lemma}
\label{lemma13}
Suppose that $\mc{T}_K = \emptyset$ and $K_c$ is a path.
Then, we can construct a $k$-pp of $K$ with at least $r |F_K| + (5-6r) |\mc{O}_K|$ edges.
\end{lemma}
\begin{proof}
If $K_c$ has no $1$-anchors, then $K$ has no satellite elements and hence we are done by Statement~1 of Lemma~\ref{lemma10}. 
So, we may assume that $K_c$ has at least one $1$-anchor. 
To obtain a $k$-pp of $K$, we construct a graph $K'$ from $K$ by initializing $i=1$ and then proceeding as follows:
\begin{itemize}
\item[P1.] If $i = \ell$ or none of $v_i, \ldots, v_\ell$ is a $1$-anchor, then set $j = \ell$ and stop.

\item[P2.] If $v_i$ is a $1$-anchor and $C_i$ is a $5$-cycle, then set $j = i$, 
remove the edge $\{ v_j, v_{j+1} \}$ if it is present, and further proceed to Step~P5.

\item[P3.] If $v_i$ is a $1$-anchor and $C_i$ is a $4$-cycle, then set $j = i+1$, 
remove the edge $\{ v_j, v_{j+1} \}$ if it is present, and further proceed to Step~P5.

\item[P4.] If $v_i$ is a $0$-anchor, then set $j$ to be the smallest index such that $v_j$ is a $1$-anchor,
remove the edge $\{ v_j, v_{j+1} \}$ if it is present.

\item[P5.] If $j = \ell$, then stop; otherwise, reset $i = j+1$ and go back to Step~P1.
\end{itemize}
We use the type-$3$ structure in the top rightmost of Figure~\ref{fig04} as an example.
If $K$ is of type-$3$, then $K'$ consists of $K[1, 2]$ and $K[3, 4]$.
Let $\Gamma$ be the set of connected components of $K'$. 
Clearly, 
$|F_{K'}| = |F_K|-(|\Gamma|-1)$.
Moreover, each $H \in \Gamma$ is $K[i, j]$ for some $1 \le i \le j \le \ell$ and $j$ is found in one of Steps~P1 through~P4. 
$H$ may not be a connected component of $\mc{F}+\mc{W}$. 
However, when $H$ has the same structure specified in an aforementioned lemma, 
we can still apply that lemma to $H$ and obtain the same result.

\medskip
{\em Case~1.} We found $j$ in Step~P2.
Then, $H = K[i]$ and $|\mc{O}_H| = 1$. 
Since $C_i$ is a $5$-cycle, Statement~2 in Lemma~\ref{lemma10} ensures that $H$ has a $k$-pp with at least $r |F_H|+r+(5-6r) |\mc{O}_H|$ edges.

\medskip
{\em Case~2.} We found $j$ in Step~P3. 
Then, $H = K[i, i+1]$. 
If $v_{i+1}$ is a $0$-anchor or $C_{i+1}$ is a $4$-cycle, then $H$ is of type-$4$.
Otherwise, $C_{i+1}$ is a $5$-cycle and hence $H$ is of type-$1$. 
In any case, Statements~1 and~2 in Lemma~\ref{lemma12} and Statement~1 in Lemma~\ref{lemma11} ensure that 
$H$ has a $k$-pp with at least $r |F_H|+r+(5-6r) |\mc{O}_H|$ edges.

\medskip
{\em Case~3.} We found $j$ in Step~P4. 
In this case, $\mc{O}_H = \{v_j\}$; in other words, $H$ is isomorphic to the leftmost graph in the bottom row of Figure~\ref{fig04}. 
Thus, by Statement~1 in Lemma~\ref{lemma12}, $H$ has a $k$-pp with at least $r |F_H|+r+(5-6r) |\mc{O}_H|$ edges.

\medskip
Let $H_\ell \in \Gamma$ be the connected component containing $v_\ell$.
For each $H = K[i, j]$ in $\Gamma$ other than $H_\ell$, $j$ is found in Step~P2, P3, or P4 and 
in turn the discussion in Case~1, 2, or~3 
shows that $H$ has a $k$-pp with at least $r |F_H|+r+(5-6r) |\mc{O}_H|$ edges.
We claim that $H_\ell$ has a $k$-pp with at least $r |F_{H_\ell}| + (5-6r)|\mc{O}_{H_\ell}|$ edges.
The claim is clearly true if the construction of $K'$ is finished in Step~P5. 
So, we may assume that the construction of $K'$ is finished in Step~P1. 
Then, either $H_\ell$ is a path and $|\mc{O}_{H_\ell}| = 0$, or $H_\ell = K[\ell]$ and $v_\ell$ is a $1$-anchor. 
In the former case, Statement~1 of Lemma~\ref{lemma10} ensures that
$H_\ell$ has a $k$-pp with at least $r |F_{H_\ell}| + (5-6r)|\mc{O}_{H_\ell}|$ edges.
In the latter case, let $t \in \{ 4, 5 \}$ be the order of $C_\ell$. 
Then, $|F_{H_\ell}| = t$, $|\mc{O}_{H_\ell}| = 1$, and 
$H_\ell$ contains a $(t+1)$-path which alone forms a $k$-pp of $H_\ell$ with $t$ edges.
By $r \ge \frac 12$, $t \ge r |F_{H_\ell}| + (5-6r)|\mc{O}_{H_\ell}|$ and hence 
$H_\ell$ has a $k$-pp with at least $r |F_{H_\ell}| + (5-6r)|\mc{O}_{H_\ell}|$ edges. 
This completes the proof of the claim.

Now, recall that $|F_{K'}| = |F_K|-(|\Gamma|-1)$, $|\mc{O}_{K'}| = |\mc{O}_K|$, 
$|F_{K'}| = \sum_{H \in \Gamma} |F_H|$, and $|\mc{O}_{K'}| = \sum_{H \in \Gamma} |\mc{O}_H|$.
By combining the $k$-pps of all $H \in \Gamma$ described in the last paragraph,
we obtain a $k$-pp of $K$ with at least 
\begin{eqnarray*}
&   & r |F_{H_\ell}| + (5-6r)|\mc{O}_{H_\ell}| + \sum_{H \in \Gamma \setminus \{ H_\ell \}} (r|F_H|+r+(5-6r)|\mc{O}_H|) \\
& = & r |F_{K'}| + (5-6r) |\mc{O}_{K'}| + r(|\Gamma|-1) \\
& = & r |F_K| + (5-6r)|\mc{O}_K|
\end{eqnarray*}
edges. This completes the proof.
\end{proof}

\begin{lemma}
\label{lemma14}
Suppose that $K$ is non-critical such that $\mc{T}_K = \emptyset$ and $K_c$ is a cycle.
Then, we construct a $k$-pp of $K$ with at least $r |F_K| + (5-6r) |\mc{O}_K|$ edges.
\end{lemma}
\begin{proof}
If $K$ has no $1$-anchors or no satellite element of $K$ is a $4$-cycle, then we are done by Statement~1 in Lemma~\ref{lemma10}.
So, we may assume that at least one satellite element of $K$ is a $4$-cycle.
Recall that $\ell = |V(K_c)| \ge 4$.
We prove the lemma by distinguishing five cases as follows.

\medskip
{\em Case~1.} There exists an index $i$ such that $C_i$ is a $4$-cycle and $C_{i+1}$ is a $5$-cycle.
Without loss of generality, we can assume $i = 1$ by reindexing the anchors in $K_c$ if necessary.
So, $H = K[1, 2]$ is isomorphic to a type-$1$ connected component of $\mc{F}+\mc{W}$.
By Statement~1 of Lemma~\ref{lemma11}, $H$ has a $k$-pp with at least $r |F_H| + 2r + (5-6r) |\mc{O}_H|$ edges.
Let $H' = K[3, \ell]$.
Since the center element of $H'$ is a path, similarly to Lemma~\ref{lemma13}, we can construct a $k$-pp for $H'$ with $r |F_{H'}| + (5-6r) |\mc{O}_{H'}|$ edges.
Note that $|F_K| = |F_H| + |F_{H'}| + 2$ and $|\mc{O}_K| = |\mc{O}_H| + |\mc{O}_{H'}|$.
By combining the $k$-pps of $H$ and $H'$, we obtain a $k$-pp of $K$ with at least 
\[
r |F_H| + 2r + (5-6r) |\mc{O}_H| + r|F_{H'}| + (5-6r) |\mc{O}_{H'}|
= r |F_K| + (5-6r) |\mc{O}_K|
\]
edges and we are done.

\medskip
{\em Case~2.} There exists an index $i$ such that both $C_i$ and $C_{i+2}$ are $4$-cycles and $v_{i+1}$ is a $0$-anchor.
Without loss of generality, we may assume $i = 1$. 
Then, $H = K[1, 3]$ is isomorphic to a type-$2$ connected component of $\mc{F}+\mc{W}$. 
Let $H' = K[4, \ell]$. 
By Statement~1 of Lemma~\ref{lemma11} and Lemma~\ref{lemma13}, we can finish the proof as in Case~1.

\medskip
{\em Case~3.} There exists an index $i$ such that both $C_i$ and $C_{i+1}$ are $4$-cycles, $C_{i+3}$ is a $5$-cycle, and $v_{i+2}$ is a $0$-anchor.  
Without loss of generality, we may assume $i = 1$. 
Then, $H = K[1, 4]$ is isomorphic to a type-$3$ connected component of $\mc{F}+\mc{W}$.
By Statement~1 in Lemma~\ref{lemma11}, $H$ contains a $k$-pp with at least $r |F_H| + 2r + (5-6r) |\mc{O}_H|$ edges.
If $K_c$ is a $4$-cycle, then $|F_K| = |F_H| + 1$ and $|\mc{O}_H| = |\mc{O}_K|$, from which the lemma follows clearly.
Otherwise, $K_c$ is a $5^+$-cycle and we can finish the proof as in Case~1
by setting $H' = K[5, \ell]$.

\medskip
{\em Case~4.} There exists an index $i$ such that $C_i$ through $C_{i+4}$ are all $4$-cycles.
Without loss of generality, we may assume $i = 1$. 
Then, $H = K[1, 4]$ is isomorphic to a type-$4$ connected component of $\mc{F}+\mc{W}$.
By Statement~2 in Lemma~\ref{lemma11} and Lemma~\ref{lemma13}, we can finish the proof as in Case~3.

\medskip
{\em Case~5.} None of Cases~1 through~4 occurs. 
We claim that $K$ has at least one $0$-anchor. 
Towards a contradiction, assume that each anchor in $K$ is a $1$-anchor. 
Recall that at least one satellite element of $K$, say $C_i$, is a $4$-cycle.
Since Case~1 does not happen, $C_{i+1}$ must be a $4$-cycle. 
By repeating this discussion, we see that every satellite element of $K$ is a $4$-cycle.
Now, since $K_c$ is a $4^+$-cycle, Case~4 occurs, a contradiction. 
So, the claim holds.

By the above claim, we may assume that $v_\ell$ is a $0$-anchor and $v_1$ is a $1$-anchor (if necessary, we reindex the vertices in $K_c$).
To obtain a $k$-pp of $K$, we construct a graph $K'$ from $K$ by initializing $i=1$ and 
then proceeding as follows always maitaining the invariant that $v_i$ is a $1$-anchor):

\begin{itemize}
\item[P1.] If either $C_i$ is a $5$-cycle, or $C_i$ is a $4$-cycle and $v_{i+1}$ is a $0$-anchor, 
then set $j$ to be the largest index such that $v_j$ is a $0$-anchor, 
remove the edge $\{ v_j, v_{j+1} \}$, and further proceed to Step~P3.

\item[P2.] If $C_i$ is a $4$-cycle and $v_{i+1}$ is a $1$-anchor, then $C_{i+1}$ is a $4$-cycle too 
(or else Case~1 would have occurred) and $i+2 \le \ell$ (because $v_\ell$ is a $0$-anchor); 
perform the following steps:

\begin{itemize}
\item[P2.1.] If $v_{i+2}$ is a $1$-anchor, then set $j = i+1$, 
remove the edge $\{ v_j, v_{j+1} \}$, and further proceed to Step~P3.

\item[P2.2.] If $v_{i+2}$ is a $0$-anchor, then 
$i+3 \le \ell$ (because otherwise, $i+2=\ell$ and $v_{i+3}=v_1$ is a $1$-anchor, 
and in turn Case~2 or~3 would have occurred depending on whether $C_1$ is a 4-cycle or 5-cycle), and 
$v_{i+3}$ is a $0$-anchor (because otherwise, Case~2 or~3 would have occurred depending on whether $C_{i+3}$ is a 4-cycle or 5-cycle);
set $j$ to be the largest index such that $v_j$ is a $0$-anchor, remove the edge $\{ v_j, v_{j+1} \}$, and further proceed to Step~P3.
\end{itemize}

\item[P3.] If $j = \ell$, then stop; otherwise, reset $i = j+1$ and go back to Step~P1.
\end{itemize}

Let $\Gamma$ be the set of connected components in $K'$.
Clearly, $|F_{K'}| = |F_K|-|\Gamma|$.
For each $H \in \Gamma$, $H$ is $K[i, j]$ for some $1 \le i \le j \le \ell$.
If $H$ is obtained in Step~P1, then by Statement~2 in Lemma~\ref{lemma10} and Statement~1 in Lemma~\ref{lemma12}, 
$H$ has a $k$-pp with at least $r |F_H| + r + (5-6r) |\mc{O}_H|$ edges. 
Similarly, if $H$ is obtained in Step~P2.1 or~P2.2, then by Statement~2 in Lemma~\ref{lemma12}, 
$H$ has a $k$-pp with at least $r |F_H| + r + (5-6r) |\mc{O}_H|$ edges. 
By combining these $k$-pps of all $H\in\Gamma$, we obtain a $k$-pp of $K$ 
with $\sum_{H \in \Gamma} (r |F_H| + r + (5-6r) |\mc{O}_H|)$ edges.
Since $|F_{K'}| = |F_K|-|\Gamma|$, $|\mc{O}_{K'}| = |\mc{O}_K|$, $|F_{K'}| = \sum_{H \in \Gamma} |F_H|$
and $|\mc{O}_{K'}| = \sum_{H \in \Gamma} |\mc{O}_H|$, 
we have
\[
\sum_{H \in \Gamma} (r |F_H| + r + (5-6r) |\mc{O}_H|) 
= r |F_{K'}| + r|\Gamma| + (5-6r) |\mc{O}_{K'}| 
= r |F_K| + (5-6r) |\mc{O}_K|.
\]
The proof is completed.
\end{proof}

\subsection{Three operations to modify $\mc{W}$}

From Lemmas~\ref{lemma13} and~\ref{lemma14}, we know that if $K$ has no $2$-anchors, i.e., $\mc{T}_K = \emptyset$,
then we can always find a good $k$-pp of $K$.
Therefore, in this subsection, we present three operations to decrease the number of $2$-anchors in $K$.

\begin{definition}
\label{def06}
$K$ is {\em balanced} if $K_c$ is a short cycle, exactly one vertex of $K_c$ is a $2$-anchor in $K$, and 
all the other vertices of $K_c$ are 0-anchors in $K$, while $K$ is {\em unbalanced} otherwise.
See Figure~\ref{fig05} for an illustration.
\end{definition}

\input{fig05}

\begin{lemma}
\label{lemma15}
Suppose $K$ is balanced. Then, it has a $k$-pp with at least $r |F_K| + 5-6r$ edges.
\end{lemma}
\begin{proof}
Without loss of generality, we may assume that $v_1$ is a $2$-anchor. 
Let $j_0$, $j_1$ and $j'_1$ be the number of vertices in $K_c$, $C_1$, and $C'_1$, respectively. 
Then, $\{j_0, j_1, j'_1\}\subseteq \{ 4, 5 \}$ and $|F_K| = j_0+j_1+j'_1$.
Obviously, $K[1]$ can be turned into a $(j_1+j'_1+1)$-path and $\bar{K}[1]$ is a $(j_0-1)$-path.
Since $j_1+j'_1+1 \le 11$ and $j_0 \le 5$, we can obtain a $k$-pp of $K$ with $j_0+j_1+j'_1-2$ edges 
by combining the $(j_1+j'_1+1)$-path and the $(j_0-1)$-path. 
Now, since $j_0+j_1+j'_1 \ge 12$ and $r < \frac 56$, we have
\[
j_0+j_1+j'_1-2 \ge \frac 56 (j_0+j_1+j'_1) \ge r (j_0+j_1+j'_1) + 10-12r > r |F_K| + 5-6r.
\]
This completes the proof.
\end{proof}

By Lemma~\ref{lemma15}, we only need to worry about unbalanced connected components of $\mc{F}+\mc{W}$ 
when defining operations for modifying $\mc{W}$.


\begin{operation}
\label{op01}
Suppose that $K$ is unbalanced with a 2-anchor $v_i$ and $G$ has an edge $\{ w, v \} \notin E(\mc{F})\cup\mc{W}$ 
such that $w\in V(C_i)$ and $v$ is a $0$-anchor. 
Then, the operation updates $\mc{W}$ by replacing $\varepsilon(C_i)$ with the edge $\{ w, v \}$.
(See Figure~\ref{fig06}.)
\end{operation}

\input{fig06}

After Operation~\ref{op01}, $\mc{W}$ is still a stingy maximum-weight path-cycle cover 
because $C_i$ is saturated by the edge $\{ w, v \}$ now. 
Obviously, Operation~\ref{op01} decreases the number of $2$-anchors 
in unbalanced connected components of $\mc{F}+\mc{W}$ by~1.

In the next definition, it is possible that $K = \hat{K}$ and $\hat{C}_j = C'_i$.

\begin{operation}
\label{op02}
Suppose that $K$ and $\hat{K}$ are connected components of $\mc{F}+\mc{W}$ such that 
$K$ is unbalanced with a 2-anchor $v_i$, 
$(\hat{K})_m$ is not an edge or $\hat{K}_c$ is not a short cycle, and $\hat{K}$ has a satellite element 
$\hat{C}_j \ne C_i$ containing a vertex $v$ adjacent to a vertex $w$ of $C_i$ in $G$. 
Then, the operation updates $\mc{W}$ by replacing 
$\varepsilon(C_i)$ and $\varepsilon(\hat{C}_j)$ with the edge $\{ w, v \}$.
\end{operation}

\input{fig07}

We need to analyze the effect of Operation~\ref{op02}. 
Obviously, Operation~\ref{op02} decreases the number of $2$-anchors in unbalanced connected components of $\mc{F}+\mc{W}$ by at least~1.
We claim that if $C$ is a short cycle saturated by $\mc{W}$ before Operation~2, then it remains saturated by $\mc{W}$ 
after Operation~2. 
Towards a contradiction, assume that the claim fails for some $C$. 
Then, it is clear that $C$ is either $K_c$ or $\hat{K}_c$. 
If $C$ is $K_c$, then $K = \hat{K}$ and $C'_i = \hat{C}_j$, and hence $K$ is balanced, a contradiction. 
So, $C$ is $\hat{K}_c$. 
But then, $(\hat{K})_c$ is an edge and $\hat{K}_c$ is a short cycle, a contradiction. 
So, the claim holds.

\begin{operation}
\label{op03}
Suppose that neither Operation~1 nor Operation~2 is applicable but $K$ and $\hat{K}$ are connected components of $\mc{F}+\mc{W}$ such that 
$K$ is unbalanced with a 2-anchor $v_i$, $(\hat{K})_m$ is an edge, $\hat{K}_c$ is a short cycle,  and $\hat{K}$ has 
a satellite element $\hat{C}_j \ne C_i$ containing a vertex $v$ adjacent to a vertex $w$ of $C_i$ in $G$. 
Then, the operation updates $\mc{W}$ by replacing $\varepsilon(C_i)$ with $\{ w, v \}$.
\end{operation}

\input{fig08}

Since $(K)_m$ is not an edge, $K \ne \hat{K}$.
Note that after Operation~3, $\hat{C}_j$ becomes the center element of the augmented $\hat{K}$. 
Clearly, if $C$ is a short cycle saturated by $\mc{W}$ before the operation, then it remains so after the operation. Moreover, if $v \ne \nu(\hat{C}_j)$, then the operation clearly decreases the number of $2$-anchors in unbalanced connected components of $\mc{F}+\mc{W}$ by~$1$. 
On the other hand, if $v = \nu(\hat{C}_j)$, then after the operation, the augmented 
$\hat{K}$ becomes balanced and hence the operation decreases the number of $2$-anchors in 
unbalanced connected components of $\mc{F}+\mc{W}$ by~$1$.

\begin{lemma}
\label{lemma16}
Operations~\ref{op01}, \ref{op02}, and \ref{op03} can be performed at most $n$ times in total.
\end{lemma}
\begin{proof}
For each $i\in\{1,2,3\}$, performing Operation~$i$ decreases the number of $2$-anchors in unbalanced connected 
components of $\mc{F}+\mc{W}$ by at least~1. Therefore, the operations can be performed at most $n$ times in total.
\end{proof}

\begin{lemma}
\label{lemma17}
Suppose that none of Operations~1, 2, and 3 is applicable but an unbalanced connected component $K$ of 
$\mc{F}+\mc{W}$ has a satellite element $C_i$ such that $\varepsilon(C_i)$ connectes $\nu(C_i)$ to 
a $2$-anchor in $K$ and $G$ has an edge $\{ w, v \} \not\in E(\mc{F}) \cup \mc{W}$ with $w\in V(C_i)$. 
Then, either $C_i$ contains both $v$ and $w$, or $v$ is a $1$- or $2$-anchor.
\end{lemma}
\begin{proof}
Towards a contradiction, assume that $v \not\in V(C_i)$ and $v$ is not a $1$- or $2$-anchor. 
Then, either $v$ is a $0$-anchor or $v$ appears in a satellite element $\hat{C}_j$ of a connected component $\hat{K}$ of $\mc{F}+\mc{W}$ with $\hat{C}_j \ne C_i$.
In the former case, Operation~1 would have been applicable, a contradiction. 
In the latter case, Operation~2 or~3 would have been applicable, a contradiction.
The lemma is proved.
\end{proof}

By Lemma~\ref{lemma16}, we repeatedly apply one of the three operations to update $\mc{W}$ until none of them is applicable.
We remark that $\mc{W}$ is always a stingy maximum-weight path-cycle cover of $G'$.

\subsection{The complete algorithm}

We inherit Notation~\ref{nota01}. 
In addition, the following notations are necessary for both the algorithm description and its analysis:

\begin{itemize}
\item $b_1$: The number of $2$-anchors $v_i$ in unbalanced connected components $K$ of $\mc{F}+\mc{W}$
 such that both $C_i$ and $C'_i$ are $4$-cycles.

\item $b_2$: The number of $2$-anchors $v_i$ in unbalanced connected components $K$ of $\mc{F}+\mc{W}$ 
such that $C_i$ or $C'_i$ is a $5$-cycle. 
We always assume that $C_i$ is a $5$-cycle.

\item $g_1$: The number of $1$-anchors in $\mc{F}+\mc{W}$.

\item $g_2$: The number of $2$-anchors in balanced connected components $K$ of $\mc{F}+\mc{W}$.

\item $V_b$: The set of all $2$-anchors $v_i$ in unbalanced connected components $K$ of $\mc{F}+\mc{W}$
together with the vertices in $V(C_i) \cup V(C'_i)$.

\item $G_b = G[V \setminus V_b]$, i.e., the graph induced by $V \setminus V_b$.
\end{itemize}

By Lemmas~\ref{lemma13}, \ref{lemma14}, and \ref{lemma15}, only for unbalanced connected components $K$
of $\mc{F}+\mc{W}$ with at least one $2$-anchor, we have to settle for a bad $k$-pp of $K$. 
If the total number of $2$-anchors in unbalanced connected components of $\mc{F}+\mc{W}$ is relatively small 
compared to the total number of $1$- or $2$-anchors in connected components of $\mc{F}+\mc{W}$, then 
the bad $k$-pps over all such $K$ will only have a small effect on the approximation ratio of our algorithm. 
Otherwise, the bad $k$-pps over all such $K$ are unacceptable and we have to find another way to deal with 
the $2$-anchors in unbalanced connected components of $\mc{F}+\mc{W}$. 
Our idea is to use recursion in this case.

Recall that $r = \frac {9-\sqrt{11}}7$ is the smaller root of the equation $7x^2-18x+10=0$.
The whole algorithm is presented in Figure~\ref{fig09}.
To analyze its approximation ratio, we distinguish three cases: 1) $2b_1+b_2 = 0$; 
2) $2b_1+b_2 \ne 0$ and $\frac {g_1+g_2}{2b_1+b_2} \ge \frac {2-2r}r$;
3) $2b_1+b_2 \ne 0$ and $\frac {g_1+g_2}{2b_1+b_2} < \frac {2-2r}r$.

\begin{figure}[htb]
\begin{center}
\framebox{
\begin{minipage}{5.3in}
Algorithm {\sc Approx2}:\\
Input: A graph $G = (V, E)$.
\begin{itemize}
\parskip=0pt
\item[1.]
         If $V$ contains exactly one vertex, then return the vertex as a $1$-path.

\item[2.]
	Call {\sc Approx1} to obtain the graph $\mc{F}+\mc{W}$.

\item[3.]
         Apply Operations~\ref{op01}, \ref{op02}, and \ref{op03} to modify $\mc{W}$ until none of them is applicable.

\item[4.] 
         If $2b_1+b_2=0$ or $\frac {g_1+g_2}{2b_1+b_2} \ge \frac {2-2r}r$,
         then perform the following steps:
  \begin{itemize}
  \item[4.1.] For each unbalanced connected component $K$ of $\mc{F}+\mc{W}$, 
	  remove $\varepsilon(C_i)$ from $K$ for every $2$-anchor $v_i$ in $K$.
  \item[4.2.] For each connected component $K$ of $\mc{F}+\mc{W}$, if $K$ is not a short cycle, 
         then construct a $k$-pp of $K$ as in Lemmas~\ref{lemma13}, \ref{lemma14}, and \ref{lemma15}; 
         otherwise, construct a $k$-pp of $K$ by removing an arbitrary edge from $K$. 
  \item[4.3.] Return the union of the $k$-pps computed in Step~4.2.
  \end{itemize}

\item[5.] 
        If $2b_1+b_2 \ne 0$  and $\frac {g_1+g_2}{2b_1+b_2} < \frac {2-2r}r$, then 
         perform the following steps:
  \begin{itemize}
  \item[5.1.] Recursively call the algorithm on $G_b$ to obtain a $k$-pp of $G_b$.
  \item[5.2.] For each unbalanced connected component $K$ of $\mc{F}+\mc{W}$ and for each 2-anchor $v_i$ in $K$, 
           obtain a $k$-pp of $G[\{v_i\}\cup V(C_i) \cup V(C'_i)]$ by removing one edge of $C_i$ incident to
           $\nu(C_i)$ and removing one edge of $C'_i$ incident to $\nu(C'_i)$.
  \item[5.3.] Return the union of the $k$-pps computed in Steps~5.2 and~5.3.
  \end{itemize}
\end{itemize}
\end{minipage}}
\end{center}
\caption{The whole approximation algorithm.} \label{fig09}
\end{figure}

Let $\mc{P}$ be the $k$-pp of $G$ returned by {\sc Approx2}.

\begin{lemma}
\label{lemma18}
Suppose $2b_1+b_2=0$. Then, $|E(\mc{P})| \ge r |E(\mc{Q})|$.
\end{lemma}
\begin{proof} 
Let $K$ be a connected component of $\mc{F}+\mc{W}$ immediately before Step~4. 
If $K$ is a short cycle, then $|F_K|-1$ edges of $K$ are included in $\mc{P}$ in Step~4.2. 
Otherwise, since $b_1 = b_2 = 0$, either $K$ is balanced, or unbalanced but with no 2-anchors.
In either case, by Lemmas~\ref{lemma13}, \ref{lemma14}, and \ref{lemma15} and the inequality $r < \frac 56$, 
at least $r |F_K|$ edges are included in $\mc{P}$ in Step~4.2.
Therefore, the number of edges in $\mc{P}$ is at least 
\[ 
r \sum_{K \notin \mc{I}_4 \cup \mc{I}_5} |F_K| + 3 |\mc{I}_4| + 4 |\mc{I}_5| 
= r |E(\mc{F})| + (3-4r) |\mc{I}_4| + (4-5r) |\mc{I}_5|.
\]
Now, the proof is finished by applying Lemma~\ref{lemma05} with $\eta=1$ and using the inequality $r < 1$.
\end{proof}

\begin{lemma}
\label{lemma19}
Suppose $2b_1+b_2 \ne 0$ and $\frac {g_1+g_2}{2b_1+b_2} \ge \frac {2-2r}r$. Then, $|E(\mc{P})| \ge r |E(\mc{Q})|$.
\end{lemma}
\begin{proof}
Immediately after Step 4.1, $E(\mc{F})$ can be partitioned into $F_1 \cup F_2$, where 
$F_1$ consists of the edges in connected components of $\mc{F}+\mc{W}$ that are short cycles, 
while $F_2 = E(\mc{F}) \setminus F_1$. 

We first discuss the edges in $F_1$. 
Immediately after Step 4.1, exactly $|\mc{I}_4|+b_1$ connected components of 
$\mc{F}+\mc{W}$ are 4-cycles and exactly $|\mc{I}_5| + b_2$ connected components of $\mc{F}+\mc{W}$ are $5$-cycles.
For these cycles, exactly $3|\mc{I}_4|+3b_1+ 4|\mc{I}_5| +4b_2$ edges are included in $\mc{P}$ in Step~4.2.

We next discuss the edges in $F_2$. 
Clearly, $|F_2| = |E(\mc{F})| - 4|\mc{I}_4|-4b_1-5|\mc{I}_5| -5b_2$.
Immediately after Step 4.1, each connected component of $\mc{F}+\mc{W}$ that is not a short cycle is 
either balanced, or unbalanced but with no $2$-anchors. 
The total number of 1-anchors in such unbalanced components is clearly $g_1+b_1+b_2$ and there are $g_2$ balanced components.
Moreover, by Lemmas~\ref{lemma13}, \ref{lemma14}, and \ref{lemma15}, 
we put at least
\[
r( |E(\mc{F})| - 4|\mc{I}_4|-4b_1-5|\mc{I}_5| -5b_2)
+ (5-6r)(g_1+g_2+b_1+b_2)
\]
edges of such connected components into $\mc{P}$ in Step~4.2.
Since $7r^2-18r+10=0$, we have $2b_1+b_2 \le \frac r{2-2r} (g_1+g_2) = \frac {5-6r}{5r-4} (g_1+g_2)$. 
Now, applying Lemma~\ref{lemma05} with $\eta=1$ and using the inequality $r < 1$, 
we see that $|E(\mc{P})|$ is at least 
\begin{eqnarray*}
&  & r|E(\mc{F})| + (3-4r) |\mc{I}_4| + (4-5r) |\mc{I}_5|
+ (5-6r)(g_1+g_2) + (8-10r)b_1 + (9-11r)b_2 \\
& \ge & r|E(\mc{Q})| + (4-4r) |\mc{I}_4| + (5-5r) |\mc{I}_5| + (5-6r)(g_1+g_2) + (8-10r)b_1 + (9-11r)b_2 \\
& \ge & r|E(\mc{Q})| + (5-6r)(g_1+g_2) + (8-10r)b_1 + (9-11r)b_2 \\
& \ge & r|E(\mc{Q})| + (5r-4) (2b_1+b_2) + (8-10r)b_1 + (9-11r)b_2 \\
& = &  r |E(\mc{Q})|  + (5-6r) b_2.
\end{eqnarray*}
Since $r < \frac 56$, the lemma holds.
\end{proof}

\begin{lemma}
\label{lemma20}
Suppose that $2b_1 + b_2\ne 0$ and $\frac {g_1+g_2}{2b_1+b_2} < \frac {2-2r}r$.
Then, $G_b$ has a $k$-pp with at least $|E(\mc{Q})| - (8b_1+10b_2+2g_1+2g_2)$ edges.
\end{lemma}
\begin{proof}
Recall that $\mc{Q}$ is an optimal $k$-pp of $G$. 
Let $\mc{Q}_b$ be the $k$-pp of $G_b$ obtained from $\mc{Q}$ by removing
all vertices in $V_b$ and the edges incident to them.
It suffices to show that $E(\mc{Q}) \setminus E(\mc{Q}_b)$ contains at most $8b_1+10b_2+2g_1+2g_2$ edges.

$E(\mc{Q}) \setminus E(\mc{Q}_b)$ can be partitioned into $X_1 \cup X_2$, where $X_1$ consists of 
all edges $\{v,w\} \in E(\mc{Q}) \setminus E(\mc{Q}_b)$ such that $v$ or $w$ is a 1- or 2-anchor, 
while $X_2$ consists of the other edges in $E(\mc{Q}) \setminus E(\mc{Q}_b)$. 
Since the total number of 1- or 2-anchors is $b_1+b_2+g_1+g_2$, $|X_1| \le 2b_1+2b_2+2g_1+2g_2$. 
To obtain an upper bound on $|X_2|$, consider an edge $\{v,w\} \in X_2$. 
Clearly, $v$ or $w$ is in $V_b$. Without loss of generality, we assume that $v \in V_b$. 
Then, since $\{v,w\}\in X_2$, neither $v$ nor $w$ is a $1$- or $2$-anchor and in turn $v$ is in
 a satellite element $C_i$ of some connected component $K$ of $\mc{F}+\mc{W}$. Hence, 
by Lemma~\ref{lemma17}, $w$ is also in $C_i$. 
Now, since a $k$-pp of $G$ can contain at most $j-1$ edges of a $j$-cycle for any $j \in \{ 4, 5 \}$, we can conclude that $|X_2| \le 6b_1+8b_2$. 
Therefore, $|X_1| + |X_2| \le 8b_1+10b_2+2g_1+2g_2$. 
This finishes the proof.
\end{proof}

\begin{theorem}
\label{thm02}
{\sc Approx2} is an $O(n^4 m^2)$-time $r$-approximation algorithm for the $k$PPE problem for every $k \ge 11$.
\end{theorem}
\begin{proof}
Recall that {\sc Approx1} runs in $O(n^3 m^2)$ time.
Note that checking whether Operation~\ref{op01}, \ref{op02}, or~\ref{op03} is applicable or not 
can be done in $O(m)$ time by checking every edge in $E(G) \setminus (E(\mc{F})\cup \mc{W})$.
By Lemma~\ref{lemma16}, Step~3 of {\sc Approx2} runs in $O(mn)$ time.
By Lemmas~\ref{lemma13} and~\ref{lemma14}, Step~4.2 of {\sc Approx2}, also takes $O(mn)$ time by finding special components defined in Definition~\ref{def05}.
Since the recursion depth of {\sc Approx2}  is at most $n$, its total running time is bounded by $O(n^4 m^2)$.

Recall that $\mc{Q}$ is an optimal $k$-pp of $G$ and $\mc{P}$ is the $k$-pp of $G$ returned by {\sc Approx2}. 
We next prove that $|E(\mc{P})| \ge r|E(\mc{Q})|$, by induction on the number of vertices in $G$. 
In the base, $G$ has exactly one vertex and {\sc Approx2} returns an optimal solution.
Next, we assume $|V| \ge 2$.
By Lemmas~\ref{lemma18} and~\ref{lemma19}, we may further assume 
 $2b_1 + b_2\ne 0$ and $\frac {g_1+g_2}{2b_1+b_2} < \frac {2-2r}r$.
By Lemma~\ref{lemma20} and the inductive hypothesis, the $k$-pp of $G_b$ computed in Step~5.1 of {\sc Approx2} contains
at least $r|E(\mc{Q})| - r(8b_1+10b_2+2g_1+2g_2)$ edges. 
Moreover, the $k$-pp of $G[\{v_i\} \cup V(C_i) \cup V(C'_i)]$ computed in Step~5.2 contains $8$ edges 
if both $C_i$ and $C'_i$ are $4$-cycles, and contains at least $9$ edges otherwise. 
So, the total number of edges included in $\mc{P}$ in Step~5.2 of {\sc Approx2} is at least $8b_1+9b_2$. 
In summary, $|E(\mc{P})| \ge r|E(\mc{Q})| - r(8b_1+10b_2+2g_1+2g_2) +8b_1+9b_2$. 
Now, by $\frac {g_1+g_2}{2b_1+b_2} < \frac {2-2r}r$ and $r < \frac 56$, we finally have 
$|E(\mc{P})| > r |E(\mc{Q})| + (5-6r)b_2 > r |E(\mc{Q})|$. 
This finishes the proof.
\end{proof}

\section{Conclusion}

In this paper, we have designed two algorithms for the problem of computing a minimum collection 
of $k^-$-paths in a given graph $G$ to cover all vertices in $G$.
One of the algorithms is a $\frac {k+4}5$-approximation algorithm for $k \in \{ 9, 10 \}$, while the other is a
$(\frac {\sqrt{11}-2}7 k + \frac {9-\sqrt{11}}7)$-approximation algorithm for $k \ge 11$.
The algorithms achieve the current best approximation ratios when $k \in \{ 9, \ldots, 18 \}$.

One interesting research direction is to extend our approaches to larger values of $k$ or directed graphs.
On the other hand, proving some inapproximability results for constant values of $k$ would be theoretically interesting, too.

\paragraph*{Acknowledgments}
This project received support from NSF award 2309902.

\bibliography{kpp.bib}

\end{document}

%% file: fig02.tex
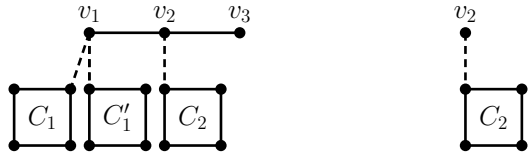
\begin{figure}[thb]
\begin{center}
\begin{tikzpicture}[scale=0.5,transform shape]

\draw [thick, line width = 1pt] (-8, -10) -- (-4, -10);
\foreach \x in {-10, -8, -6} 
{
      \draw[thick, line width = 1pt] (\x, -11.5) rectangle (\x+1.5, -13);
     \foreach \y in {-11.5, -13} {
           \fill (\x, \y) circle(.15);
           \fill (\x+1.5, \y) circle(.15); }
}
\foreach \x in {-8, -6}
{
    \fill (\x, -10) circle(.15);
    \draw[densely dashed, line width = 1pt] (\x, -10) -- (\x, -11.5);
}
\fill (-4, -10) circle(.15);

\draw [densely dashed, line width = 1pt] (-8, -10) -- (-8.5, -11.5);
\draw [densely dashed, line width = 1pt] (-8, -10) -- (-8, -11.5);

\node[font=\fontsize{20}{6}\selectfont] at (-8, -9.5) {$v_1$};
\node[font=\fontsize{20}{6}\selectfont] at (-6, -9.5) {$v_2$};
\node[font=\fontsize{20}{6}\selectfont] at (-4, -9.5) {$v_3$};
\node[font=\fontsize{20}{6}\selectfont] at (-9.25, -12.25) {$C_1$};
\node[font=\fontsize{20}{6}\selectfont] at (-7.25, -12.25) {$C'_1$};
\node[font=\fontsize{20}{6}\selectfont] at (-5.25, -12.25) {$C_2$};


\fill (2, -10) circle(.15);
\draw [densely dashed, line width = 1pt] (2, -10) -- (2, -11.5);
\draw[thick, line width = 1pt] (2, -11.5) rectangle (3.5, -13);
\foreach \y in {-11.5, -13} {
    \fill (2, \y) circle(.15);
    \fill (3.5, \y) circle(.15); }

\node[font=\fontsize{20}{6}\selectfont] at (2, -9.5) {$v_2$};
\node[font=\fontsize{20}{6}\selectfont] at (2.75, -12.25) {$C_2$};





\end{tikzpicture}
\end{center}
\caption{The left graph shows a possible structure of a connected component $K$ in $\mc{F}+\mc{W}$
while the right graph shows $K[2]$.
Each thick (respectively, dashed) edge is in $\mc{F}$ (respectively, $\mc{W}$).
$K_c$ is a $3$-path and each satellite element of $K$ is a $4$-cycle in $\mc{F}$.
Moreover, $v_i$ is a $(3-i)$-anchor for $i \in \{ 1, 2, 3 \}$.
\label{fig02}}
\end{figure}

%% file: fig03.tex
\begin{figure}[thb]
\begin{center}
\begin{tikzpicture}[scale=0.5,transform shape]

\draw[thick, line width = 1pt] (16, 0) rectangle (17.5, 1.5);
       \foreach \y in {0, 1.5} {
               \fill (16, \y) circle(.15);
               \fill (17.5, \y) circle(.15); }

\foreach \x in {14, 18} {
       \draw[thick, line width = 1pt] (\x, -1.5) rectangle (\x+1.5, -3);
       \foreach \y in {-1.5, -3} {
               \fill (\x, \y) circle(.15);
               \fill (\x+1.5, \y) circle(.15); }
}

\draw[densely dashed, line width = 1pt] (16, 0) -- (15.5, -1.5);
\draw[densely dashed, line width = 1pt] (17.5, 0) -- (18, -1.5);

\node[font=\fontsize{20}{6}\selectfont] at (15.5, 0) {$v_1$};
\node[font=\fontsize{20}{6}\selectfont] at (18, 0) {$v_2$};
\node[font=\fontsize{20}{6}\selectfont] at (18, 1.5) {$v_3$};
\node[font=\fontsize{20}{6}\selectfont] at (15.5, 1.5) {$v_4$};
\node[font=\fontsize{20}{6}\selectfont] at (18.75, -2.25) {$C_2$};
\node[font=\fontsize{20}{6}\selectfont] at (14.75, -2.25) {$C_1$};


\draw[thick, line width = 1pt] (29.5, 1.5) -- (28, 1.5);
\draw[thick, line width = 1pt] (29.5, 1.5) -- (29.5, 0);
       \foreach \y in {0, 1.5} {
               \fill (28, \y) circle(.15);
               \fill (29.5, \y) circle(.15); }

\foreach \x in {26, 30} {
       \draw[thick, line width = 1pt] (\x, -1.5) rectangle (\x+1.5, -3);
       \foreach \y in {-1.5, -3} {
               \fill (\x, \y) circle(.15);
               \fill (\x+1.5, \y) circle(.15); }
}

\draw[thick, line width = 1pt] (28, 0) -- (27.5, -1.5);
\draw[thick, line width = 1pt] (29.5, 0) -- (30, -1.5);

\node[font=\fontsize{20}{6}\selectfont] at (27.5, 0) {$v_1$};
\node[font=\fontsize{20}{6}\selectfont] at (30, 0) {$v_2$};
\node[font=\fontsize{20}{6}\selectfont] at (30, 1.5) {$v_3$};
\node[font=\fontsize{20}{6}\selectfont] at (27.5, 1.5) {$v_4$};
\node[font=\fontsize{20}{6}\selectfont] at (30.75, -2.25) {$C_2$};
\node[font=\fontsize{20}{6}\selectfont] at (26.75, -2.25) {$C_1$};

\end{tikzpicture}
\end{center}
\caption{
An example connected component $K$ of $\mc{F}+\mc{W}$ for Case~1 in the proof of Lemma~\ref{lemma08}. 
Removing the two edges $\{ v_0, v_1 \} = \{ v_4, v_1 \}, \{ v_1, v_2 \}$ yields a graph whose connected components 
are $K[1]$ and $K[2, 4]$. $K[1]$ (respectively, $K[2, 4]$) has a $k$-pp with $4$ (respectively, $6$) edges.
\label{fig03}}
\end{figure}
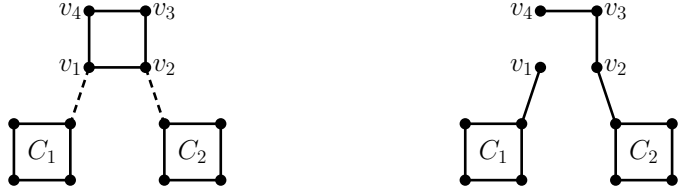

%% file: fig04.tex
\begin{figure}[thb]
\begin{center}
\begin{tikzpicture}[scale=0.5,transform shape]

\draw [thick, line width = 1pt] (-6, -10) -- (-4, -10);
\draw[thick, line width = 1pt] (-6, -11.5) rectangle (-4.5, -13);
\draw [thick, line width = 1pt] (-4, -11.5) -- (-4, -13);
\draw [thick, line width = 1pt] (-4, -11.5) -- (-2.7, -11.5);
\draw [thick, line width = 1pt] (-4, -13) -- (-2.7, -13);
\draw [thick, line width = 1pt] (-2, -12.25) -- (-2.7, -11.5);
\draw [thick, line width = 1pt] (-2, -12.25) -- (-2.7, -13);
\fill (-2, -12.25) circle(.15);

\foreach \x in {-6, -4.5, -2.7, -4} {
       \foreach \y in {-11.5, -13} {
               \fill (\x, \y) circle(.15);}
}
\foreach \x in {-6, -4}
{
    \fill (\x, -10) circle(.15);
    \draw[densely dashed, line width = 1pt] (\x, -10) -- (\x, -11.5);
}
 
\draw [densely dashed, line width = 1pt] (-6, -10) -- (-6, -11.5);

\node[font=\fontsize{20}{6}\selectfont] at (-6, -9.5) {$v_1$};
\node[font=\fontsize{20}{6}\selectfont] at (-4, -9.5) {$v_2$};
\node[font=\fontsize{20}{6}\selectfont] at (-5.25, -12.25) {$C_1$};
\node[font=\fontsize{20}{6}\selectfont] at (-3.2, -12.25) {$C_2$};


\draw [thick, line width = 1pt] (3, -10) -- (7, -10);

\foreach \x in {3, 5, 7} {
       \fill (\x, -10) circle(.15);
}

\foreach \x in {3, 7} {
       \draw[thick, line width = 1pt] (\x, -11.5) rectangle (\x+1.5, -13);
       \foreach \y in {-11.5, -13} {
               \fill (\x, \y) circle(.15);
               \fill (\x+1.5, \y) circle(.15); }
}

\foreach \x in {3, 7}
{
    \draw[densely dashed, line width = 1pt] (\x, -10) -- (\x, -11.5);
}

\node[font=\fontsize{20}{6}\selectfont] at (3, -9.5) {$v_1$};
\node[font=\fontsize{20}{6}\selectfont] at (5, -9.5) {$v_2$};
\node[font=\fontsize{20}{6}\selectfont] at (7, -9.5) {$v_3$};
\node[font=\fontsize{20}{6}\selectfont] at (3.75, -12.25) {$C_1$};
\node[font=\fontsize{20}{6}\selectfont] at (7.75, -12.25) {$C_3$};


\draw [thick, line width = 1pt] (13, -10) -- (19, -10);

\foreach \x in {13, 15, 17, 19} {
       \fill (\x, -10) circle(.15);
}

\foreach \x in {13, 15} {
       \draw[thick, line width = 1pt] (\x, -11.5) rectangle (\x+1.5, -13);
       \foreach \y in {-11.5, -13} {
               \fill (\x, \y) circle(.15);
               \fill (\x+1.5, \y) circle(.15); }
}

\foreach \x in {13, 15, 19}
{
    \draw[densely dashed, line width = 1pt] (\x, -10) -- (\x, -11.5);
}

\draw [thick, line width = 1pt] (19, -11.5) -- (19, -13);
\draw [thick, line width = 1pt] (19, -11.5) -- (20.3, -11.5);
\draw [thick, line width = 1pt] (19, -13) -- (20.3, -13);
\draw [thick, line width = 1pt] (21, -12.25) -- (20.3, -11.5);
\draw [thick, line width = 1pt] (21, -12.25) -- (20.3, -13);
\fill (21, -12.25) circle(.15);
\foreach \x in {19, 20.3} {
       \foreach \y in {-11.5, -13} {
              \fill (\x, \y) circle(.15);
}}

\node[font=\fontsize{20}{6}\selectfont] at (13, -9.5) {$v_1$};
\node[font=\fontsize{20}{6}\selectfont] at (15, -9.5) {$v_2$};
\node[font=\fontsize{20}{6}\selectfont] at (17, -9.5) {$v_3$};
\node[font=\fontsize{20}{6}\selectfont] at (19, -9.5) {$v_4$};
\node[font=\fontsize{20}{6}\selectfont] at (13.75, -12.25) {$C_1$};
\node[font=\fontsize{20}{6}\selectfont] at (15.75, -12.25) {$C_2$};
\node[font=\fontsize{20}{6}\selectfont] at (19.75, -12.25) {$C_4$};


\draw [thick, line width = 1pt] (-6, -16) -- (-4, -16);
\draw[thick, line width = 1pt] (-6, -17.5) rectangle (-4.5, -19);
\fill (-6, -16) circle(.15);
\fill (-4, -16) circle(.15);
\draw[densely dashed, line width = 1pt] (-6, -16) -- (-6, -17.5);
\foreach \y in {-17.5, -19} {
               \fill (-6, \y) circle(.15);
               \fill (-4.5, \y) circle(.15); }

\node[font=\fontsize{20}{6}\selectfont] at (-6, -15.5) {$v_1$};
\node[font=\fontsize{20}{6}\selectfont] at (-4, -15.5) {$v_2$};
\node[font=\fontsize{20}{6}\selectfont] at (-5.25, -18.25) {$C_1$};


\draw [thick, line width = 1pt] (3, -16) -- (9, -16);

\foreach \x in {3, 5, 7, 9} {
       \fill (\x, -16) circle(.15);
}

\foreach \x in {3, 5} {
       \draw[thick, line width = 1pt] (\x, -17.5) rectangle (\x+1.5, -19);
       \foreach \y in {-17.5, -19} {
               \fill (\x, \y) circle(.15);
               \fill (\x+1.5, \y) circle(.15); }
}

\foreach \x in {3, 5}
{
    \draw[densely dashed, line width = 1pt] (\x, -16) -- (\x, -17.5);
}

\node[font=\fontsize{20}{6}\selectfont] at (3, -15.5) {$v_1$};
\node[font=\fontsize{20}{6}\selectfont] at (5, -15.5) {$v_2$};
\node[font=\fontsize{20}{6}\selectfont] at (7, -15.5) {$v_3$};
\node[font=\fontsize{20}{6}\selectfont] at (9, -15.5) {$v_4$};
\node[font=\fontsize{20}{6}\selectfont] at (3.75, -18.25) {$C_1$};
\node[font=\fontsize{20}{6}\selectfont] at (5.75, -18.25) {$C_2$};


\draw [thick, line width = 1pt] (13, -16) -- (19, -16);

\foreach \x in {13, 15, 17, 19} {
       \fill (\x, -16) circle(.15);
}

\foreach \x in {13, 15, 17, 19} {
       \draw[thick, line width = 1pt] (\x, -17.5) rectangle (\x+1.5, -19);
       \foreach \y in {-17.5, -19} {
               \fill (\x, \y) circle(.15);
               \fill (\x+1.5, \y) circle(.15); }
}

\foreach \x in {13, 15, 17, 19}
{
    \draw[densely dashed, line width = 1pt] (\x, -16) -- (\x, -17.5);
}

\node[font=\fontsize{20}{6}\selectfont] at (13, -15.5) {$v_1$};
\node[font=\fontsize{20}{6}\selectfont] at (15, -15.5) {$v_2$};
\node[font=\fontsize{20}{6}\selectfont] at (17, -15.5) {$v_3$};
\node[font=\fontsize{20}{6}\selectfont] at (19, -15.5) {$v_4$};
\node[font=\fontsize{20}{6}\selectfont] at (13.75, -18.25) {$C_1$};
\node[font=\fontsize{20}{6}\selectfont] at (15.75, -18.25) {$C_2$};
\node[font=\fontsize{20}{6}\selectfont] at (17.75, -18.25) {$C_3$};
\node[font=\fontsize{20}{6}\selectfont] at (19.75, -18.25) {$C_4$};

\end{tikzpicture}
\end{center}
\caption{Possible structures of special connected components $K$ of $\mc{F}+\mc{W}$, 
where each thick (respectively, dashed) edge is in $\mc{F}$ (respectively, $\mc{W}$). 
The top row shows the structure of a type-1, type-2, or type-3 $K$
from left to right. The bottom row shows three possible structures of a type-$4$ $K$,
where the pair $(i, \ell)$ stated in Condition~S4 is $(1, 2)$, $(2, 4)$ and $(4, 4)$, respectively.}
\label{fig04}
\end{figure}
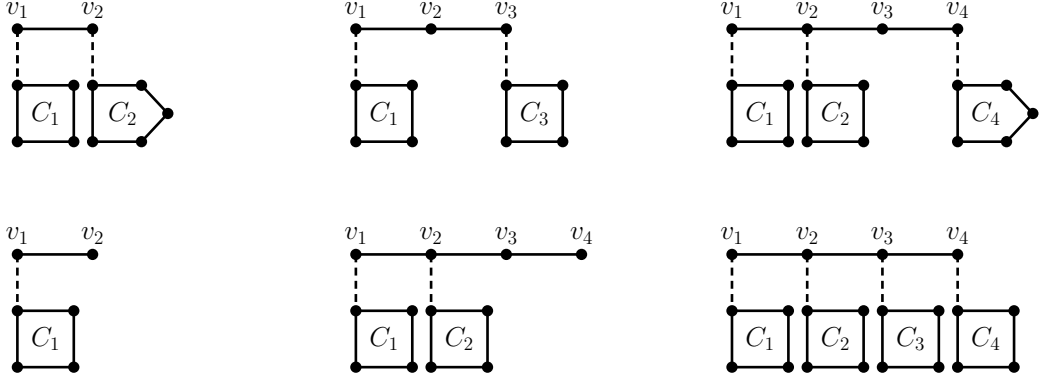

%% file: fig05.tex
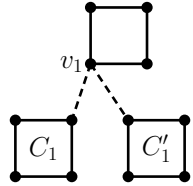
\begin{figure}[thb]
\begin{center}
\begin{tikzpicture}[scale=0.5,transform shape]

\draw[thick, line width = 1pt] (16, 0) rectangle (17.5, 1.5);
       \foreach \y in {0, 1.5} {
               \fill (16, \y) circle(.15);
               \fill (17.5, \y) circle(.15); }

\foreach \x in {14, 17} {
       \draw[thick, line width = 1pt] (\x, -1.5) rectangle (\x+1.5, -3);
       \foreach \y in {-1.5, -3} {
               \fill (\x, \y) circle(.15);
               \fill (\x+1.5, \y) circle(.15); }
}

\foreach \x in {15.5, 17}
{
    \draw[densely dashed, line width = 1pt] (16, 0) -- (\x, -1.5);
}

\node[font=\fontsize{20}{6}\selectfont] at (15.5, 0) {$v_1$};
\node[font=\fontsize{20}{6}\selectfont] at (17.75, -2.25) {$C'_1$};
\node[font=\fontsize{20}{6}\selectfont] at (14.75, -2.25) {$C_1$};

\end{tikzpicture}
\end{center}
\caption{A balanced connected component of $\mc{F}+\mc{W}$, where
each thick and dashed edge is in $\mc{F}$ and $\mc{W}$, respectively.
\label{fig05}}
\end{figure}

%% file: fig06.tex
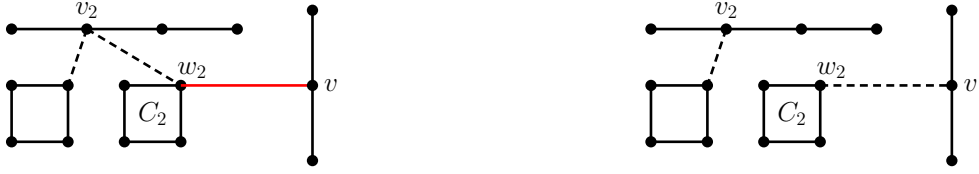
\begin{figure}[thb]
\begin{center}
\begin{tikzpicture}[scale=0.5,transform shape]

\draw [thick, line width = 1pt] (0, 0) -- (6, 0);

\foreach \x in {0, 2, 4, 6} {
       \fill (\x, 0) circle(.15);
}

\foreach \x in {0, 3} {
       \draw[thick, line width = 1pt] (\x, -1.5) rectangle (\x+1.5, -3);
       \foreach \y in {-1.5, -3} {
               \fill (\x, \y) circle(.15);
               \fill (\x+1.5, \y) circle(.15); }
}

\foreach \x in {1.5, 4.5}
{
    \draw[densely dashed, line width = 1pt] (2, 0) -- (\x, -1.5);
}
\draw[red, thick, line width = 1pt] (4.5, -1.5) -- (8, -1.5);

\node[font=\fontsize{20}{6}\selectfont] at (2, 0.5) {$v_2$};
\node[font=\fontsize{20}{6}\selectfont] at (4.8, -1.1) {$w_2$};
\node[font=\fontsize{20}{6}\selectfont] at (3.75, -2.25) {$C_2$};

\draw [thick, line width = 1pt] (8, 0.5) -- (8, -3.5);

\foreach \x in {0.5, -1.5, -3.5} {
       \fill (8, \x) circle(.15);
}

\node[font=\fontsize{20}{6}\selectfont] at (8.5, -1.5) {$v$};


\draw [thick, line width = 1pt] (17, 0) -- (23, 0);

\foreach \x in {17, 19, 21, 23} {
       \fill (\x, 0) circle(.15);
}

\foreach \x in {17, 20} {
       \draw[thick, line width = 1pt] (\x, -1.5) rectangle (\x+1.5, -3);
       \foreach \y in {-1.5, -3} {
               \fill (\x, \y) circle(.15);
               \fill (\x+1.5, \y) circle(.15); }
}

\draw[densely dashed, line width = 1pt] (19, 0) -- (18.5, -1.5);
\draw[densely dashed, line width = 1pt] (21.5, -1.5) -- (25, -1.5);

\node[font=\fontsize{20}{6}\selectfont] at (19, 0.5) {$v_2$};
\node[font=\fontsize{20}{6}\selectfont] at (21.8, -1.1) {$w_2$};
\node[font=\fontsize{20}{6}\selectfont] at (20.75, -2.25) {$C_2$};

\draw [thick, line width = 1pt] (25, 0.5) -- (25, -3.5);

\foreach \x in {0.5, -1.5, -3.5} {
       \fill (25, \x) circle(.15);
}

\node[font=\fontsize{20}{6}\selectfont] at (25.5, -1.5) {$v$};

\end{tikzpicture}
\end{center}
\caption{A representative example for Operation~\ref{op01},
where the dashed $\{ w_2, v_2 \}$ is replaced with the red edge $\{ w_2, v \}$.
\label{fig06}}
\end{figure}

%% file: fig07.tex
\begin{figure}[thb]
\begin{center}
\begin{tikzpicture}[scale=0.5,transform shape]

\draw [thick, line width = 1pt] (0, 0) -- (6, 0);

\foreach \x in {0, 2, 4, 6} {
       \fill (\x, 0) circle(.15);
}

\foreach \x in {0, 3} {
       \draw[thick, line width = 1pt] (\x, -1.5) rectangle (\x+1.5, -3);
       \foreach \y in {-1.5, -3} {
               \fill (\x, \y) circle(.15);
               \fill (\x+1.5, \y) circle(.15); }
}

\foreach \x in {1.5, 3}
{
    \draw[densely dashed, line width = 1pt] (2, 0) -- (\x, -1.5);
}
\draw[red, thick, line width = 1pt] (1.5, -1.5) -- (3, -1.5);

\node[font=\fontsize{20}{6}\selectfont] at (2, 0.5) {$v_2$};
\node[font=\fontsize{20}{6}\selectfont] at (1, -1.1) {$w_2$};
\node[font=\fontsize{20}{6}\selectfont] at (3.3, -1) {$w'_2$};
\node[font=\fontsize{20}{6}\selectfont] at (3.75, -2.25) {$C'_2$};
\node[font=\fontsize{20}{6}\selectfont] at (0.75, -2.25) {$C_2$};


\draw [thick, line width = 1pt] (14, 0) -- (20, 0);

\foreach \x in {14, 16, 18, 20} {
       \fill (\x, 0) circle(.15);
}

\foreach \x in {14, 17} {
       \draw[thick, line width = 1pt] (\x, -1.5) rectangle (\x+1.5, -3);
       \foreach \y in {-1.5, -3} {
               \fill (\x, \y) circle(.15);
               \fill (\x+1.5, \y) circle(.15); }
}

\draw[densely dashed, line width = 1pt] (15.5, -1.5) -- (17, -1.5);

\node[font=\fontsize{20}{6}\selectfont] at (16, 0.5) {$v_2$};
\node[font=\fontsize{20}{6}\selectfont] at (15, -1.1) {$w_2$};
\node[font=\fontsize{20}{6}\selectfont] at (17.3, -1) {$w'_2$};
\node[font=\fontsize{20}{6}\selectfont] at (17.75, -2.25) {$C'_2$};
\node[font=\fontsize{20}{6}\selectfont] at (14.75, -2.25) {$C_2$};

\end{tikzpicture}
\end{center}
\caption{A representative example for Operation~\ref{op02},
where the two dashed edges $\{ v_2, w_2 \}, \{ v_2, w'_2 \}$ are replaced with the red edge $\{ w_2, w'_2 \}$.
\label{fig07}}
\end{figure}
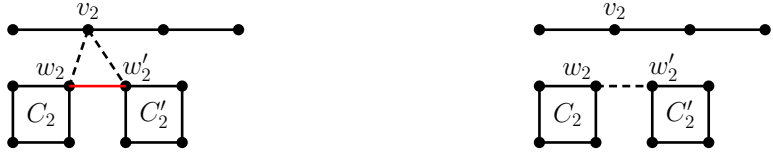

%% file: fig08.tex
\begin{figure}[thb]
\begin{center}
\begin{tikzpicture}[scale=0.5,transform shape]

\draw [thick, line width = 1pt] (0, 0) -- (6, 0);

\foreach \x in {0, 2, 4, 6} {
       \fill (\x, 0) circle(.15);
}

\foreach \x in {0, 3} {
       \draw[thick, line width = 1pt] (\x, -1.5) rectangle (\x+1.5, -3);
       \foreach \y in {-1.5, -3} {
               \fill (\x, \y) circle(.15);
               \fill (\x+1.5, \y) circle(.15); }
}

\foreach \x in {1.5, 4.5}
{
    \draw[densely dashed, line width = 1pt] (2, 0) -- (\x, -1.5);
}

\foreach \x in {1, -2}
{
    \draw[thick, line width = 1pt] (8, \x) rectangle (9.5, \x-1.5);
    \foreach \y in {8, 9.5} {
            \fill (\y, \x) circle(.15);
            \fill (\y, \x-1.5) circle(.15); }
}

\draw[red, thick, line width = 1pt] (4.5, -1.5) -- (8, -2);

\draw[densely dashed, line width = 1pt] (8, 1) -- (8, -2.5);
\node[font=\fontsize{20}{6}\selectfont] at (7.6, -2.3) {$v$};
\node[font=\fontsize{20}{6}\selectfont] at (2, 0.5) {$v_2$};
\node[font=\fontsize{20}{6}\selectfont] at (4.8, -1.1) {$w_2$};
\node[font=\fontsize{20}{6}\selectfont] at (3.75, -2.25) {$C_2$};


\draw [thick, line width = 1pt] (17, 0) -- (23, 0);

\foreach \x in {17, 19, 21, 23} {
       \fill (\x, 0) circle(.15);
}

\foreach \x in {17, 20} {
       \draw[thick, line width = 1pt] (\x, -1.5) rectangle (\x+1.5, -3);
       \foreach \y in {-1.5, -3} {
               \fill (\x, \y) circle(.15);
               \fill (\x+1.5, \y) circle(.15); }
}

\draw[densely dashed, line width = 1pt] (19, 0) -- (18.5, -1.5);

\foreach \x in {1, -2}
{
    \draw[thick, line width = 1pt] (25, \x) rectangle (26.5, \x-1.5);
    \foreach \y in {25, 26.5} {
            \fill (\y, \x) circle(.15);
            \fill (\y, \x-1.5) circle(.15); }
}

\draw[densely dashed, line width = 1pt] (21.5, -1.5) -- (25, -2);

\draw[densely dashed, line width = 1pt] (25, 1) -- (25, -2.5);
\node[font=\fontsize{20}{6}\selectfont] at (24.6, -2.3) {$v$};
\node[font=\fontsize{20}{6}\selectfont] at (19, 0.5) {$v_2$};
\node[font=\fontsize{20}{6}\selectfont] at (21.8, -1.1) {$w_2$};
\node[font=\fontsize{20}{6}\selectfont] at (20.75, -2.25) {$C_2$};

\end{tikzpicture}
\end{center}
\caption{A representative example for Operation~\ref{op03},
where the edge $\{ v_2, w_2 \}$ is replaced with the red edge $\{ w_2, v \}$.
\label{fig08}}
\end{figure}